\documentclass[peerreview,12pt,draftclsnofoot,journal,onecolumn]{IEEEtran}

\IEEEoverridecommandlockouts

\usepackage{graphicx}
\usepackage{cite}
\usepackage{url}
\usepackage[cmex10]{amsmath}
\usepackage{amssymb}
\usepackage{algorithm}
\usepackage[noend]{algorithmic}
\usepackage{cases}
\usepackage[caption=false,font=footnotesize]{subfig}
\usepackage{color}

\usepackage[nofiglist,notablist,nomarkers]{endfloat}

\newtheorem{theorem}{Theorem}

\DeclareMathOperator{\E}{\mathbb{E}}

\begin{document}

\title{Cramer-Rao Bounds for Joint RSS/DoA-Based Primary-User Localization in Cognitive Radio Networks}
\author{Jun~Wang, Jianshu~Chen, and Danijela~Cabric%
\thanks{The authors are with the Department of Electrical Engineering, University of California, Los Angeles, CA, 90095, USA (email: \{eejwang, jshchen, danijela\}@ee.ucla.edu).}
}

\maketitle


\begin{abstract}
Knowledge about the location of licensed primary-users (PU) could enable several key features in cognitive radio (CR) networks including improved spatio-temporal sensing, intelligent location-aware routing, as well as aiding spectrum policy enforcement. In this paper we consider the achievable accuracy of PU localization algorithms that jointly utilize received-signal-strength (RSS) and direction-of-arrival (DoA) measurements by evaluating the Cramer-Rao Bound (CRB). Previous works evaluate the CRB for RSS-only and DoA-only localization algorithms separately and assume DoA estimation error variance is a fixed constant or rather independent of RSS. We derive the CRB for joint RSS/DoA-based PU localization algorithms based on the mathematical model of DoA estimation error variance as a function of RSS, for a given CR placement. The bound is compared with practical localization algorithms and the impact of several key parameters, such as number of nodes, number of antennas and samples, channel shadowing variance and correlation distance, on the achievable accuracy are thoroughly analyzed and discussed. We also derive the closed-form asymptotic CRB for uniform random CR placement, and perform theoretical and numerical studies on the required number of CRs such that the asymptotic CRB tightly approximates the numerical integration of the CRB for a given placement.
\end{abstract}

\IEEEpeerreviewmaketitle


\section{Introduction}
\label{sec:Intro}

Cognitive Radio (CR) is a promising approach to efficiently utilize the scarce RF spectrum resources \cite{Mitola1999}. In this paradigm, knowledge about spectrum occupancy in time, frequency, and space that is both accurate and timely is crucial in allowing CR networks to opportunistically use the spectrum and avoid interference to a primary user (PU) \cite{Haykin2005}. In particular, information about PU location could enable several key capabilities in CR networks including improved spatio-temporal sensing, intelligent location-aware routing, as well as aiding spectrum policy enforcement \cite{Celebi2007a}.

The PU localization problem in CR networks is in general different from localization in other applications such as Wireless Sensor Networks (WSN) \cite{Patwari2005} and Global Positioning System (GPS) \cite{Kaplan2005}, due to the following two features. First, a PU does not communicate with CRs and very limited knowledge about its signaling, such as transmit power or modulation scheme, is available. Therefore, passive localization techniques should be applied. Second, since CRs need to detect and localize PUs in the whole coverage area at a very low signal-to-noise ratio (SNR), in order to avoid interference to the primary network, the required number of CRs is relatively large and cooperation among CRs is necessary.

\subsection{Related Work}
Prior research on passive localization can be categorized into three classes based on the types of measurements shared among sensors to obtain location estimates \cite{Patwari2005}. Received-signal-strength (RSS) based algorithms use measured received power from the PU to provide coarse-grained estimates at a low hardware and computational cost. Time-difference-of-arrival (TDoA) based algorithms obtain location estimates from time differences among multiple receptions of the transmitted signal. They are not suitable for CR applications since TDoA-based algorithms require perfect synchronization among CRs. Direction-of-arrival (DoA) based algorithms use target DoA estimates observed at different receivers to obtain location estimates. The DoA estimates can be obtained from either multiple antenna arrays, directional antennas or virtual arrays formed by cooperative CRs. Therefore, RSS and DoA-based algorithms are the proper choices for the PU localization problem.

Passive localization based on RSS and/or DoA measurements has been well-studied in both WSN and CR literature, in terms of algorithm design and evaluation. RSS-based algorithms include range-based and range-free techniques, a detailed survey can be found in \cite{Wang2011}. Popular DoA fusion algorithms include Maximum Likelihood (ML) \cite{Gavish1992}, Stansfield algorithm \cite{Stansfield1947} and Linear Least Square (LLS), each of which provides different tradeoff between accuracy and complexity. Algorithms that jointly utilize RSS and DoA information are also proposed. For example, weighted versions of the popular DoA fusion algorithms are shown to improve the localization accuracy \cite{Gavish1992}. The weights are typically estimated error variances of individual DoA measurements, that are obtained from instant RSS \cite{Stoica1989}.

In this paper, we focus on characterizing the achievable performance of joint RSS/DoA-based localization algorithms by means of the Cramer-Rao Bound (CRB), which provides a lower bound on the estimation accuracy of any unbiased estimator using either or both RSS and DoA. The CRB for RSS-only localization is studied in \cite{Patwari2005} and \cite{Patwari2008} assuming independent and correlated shadowing channel, respectively. The DoA-only CRB is discussed in several papers \cite{Patwari2005, Gustafsson2005, Ash2008, Seow2008, Vaghefi2010}, however they all assume the DoA estimates from different CRs are subject to i.i.d. Gaussian errors with zero mean and fixed variance. Since the DoA estimation error variance depends on instant RSS and other system parameters (e.g. number of antennas and samples, array orientation error) \cite{Stoica1989}, assuming it's constant makes the derived bounds less accurate. Recently, the mathematical model of DoA estimation error is used in \cite{Fu2009, Penna2011, Qi2006} to calculate DoA estimation error variance for some specific system settings, such as node density and placement, using fixed value of RSS for all CRs. The variance values are then plugged into the existing DoA-only CRB to evaluate the localization performance. This approach does not include the system parameters in the mathematical model of DoA error variance and channel parameters in RSS, thus the CRB is less comprehensive.

\subsection{Contributions}
Since the DoA estimation accuracy strongly depends on RSS, and practical algorithms show that using both RSS and DoA estimates provides better localization accuracy than using DoA alone, the CRB for joint RSS/DoA-based localization is more capable to characterize the positioning ability of the CR network. Along this direction, this paper has the following contributions:

\begin{enumerate}

\item Derivation of the CRB for joint RSS/DoA-based passive localization considering the interdependency between RSS and DoA, for a given CR placement. The CRB is derived for both optimal DoA estimation algorithm, with error given by the CRB of DoA estimation error variance, and MUSIC algorithm. Therefore, the derived CRB provides both ultimate achievable accuracy and the accuracy achieved by a practical DoA estimation algorithm.

\item Derivation of closed-form asymptotic joint RSS/DoA-based CRB for uniform CR placement with i.i.d shadowing. The asymptotic RSS-only CRB for such case is also derived as a byproduct. We also provide theoretical and numerical studies on the required number of CRs that ensures the asymptotic CRB tightly approaches the numerical integration of the CRB for a given placement.

\item Study of the impact of various system, channel and array parameters on the localization accuracy, theoretically and via simulation. The considered parameters include number of CRs, shadowing variance, correlation distance, number of antennas, number of samples and array orientation error. The study provides insights and guidelines for practical system design of PU localization in CR network.

\item Comparison of the derived CRBs with practical RSS-only and joint RSS/DoA localization algorithms. We study the robustness of practical algorithms to major system parameters, and reveal the potential improvements of practical localization system and algorithm design.

\end{enumerate}

The rest of the paper is organized as follows. The system model is introduced in Section \ref{sec:Model}. The joint CRB derivation for a fixed CR placement is presented in Section \ref{sec:JCRB_Fixed}. The asymptotic joint CRB derivation for uniform CR placement is shown in Section \ref{sec:JCRB_Asym}. Numerical results evaluating impacts of various parameters on the joint CRB are discussed in Section \ref{sec:Simulation}. Finally, the paper is concluded in Section \ref{sec:Conclusion}.


\section{System Model}
\label{sec:Model}

Assume $N$ CRs cooperate to localize a single PU. 2-dimensional locations of the PU and the $n^{\text{th}}$ CR are denoted as $\boldsymbol{\ell}_{P}=[x_{P}, y_{P}]^T$ and $\boldsymbol{\ell}_{n}=[x_{n}, y_{n}]^T$, respectively. All locations are static in the observation period, and CR locations are known. Available measurements at CRs are RSS and DoA. RSS measurement at the $n^{\text{th}}$ CR is modeled as $\widehat{\psi}_n \triangleq P_T \frac{c_0 10^{-s_n/10}}{d_n^{\gamma}} \text{Watt}$, where $P_T$ is the PU transmit power, $c_0$ is the (constant) average multiplicative gain at reference distance, $d_n = \| \boldsymbol{\ell}_{n} - \boldsymbol{\ell}_{P} \|$ is the distance between the $n^{\text{th}}$ CR and PU, $\gamma$ is the path loss exponent, and $10^{-s_n/10}$ is a random variable that reflects shadowing. The RSS is normally expressed in dBm using the transformation $\widehat{\phi}_n = 10 \log_{10} (1000 \widehat{\psi}_n)$, the result is given as
\begin{equation}\label{RSS_dB}
    \widehat{\phi}_n = 10 \log_{10} (1000 P_T c_0) - 10 \gamma \log_{10} d_n -s_n \triangleq \overline{\phi}_n - s_n, \; \text{dBm.}
\end{equation}
Denote collection of RSS measurements from all CRs as $\widehat{\boldsymbol{\phi}}=[\widehat{\phi}_1, \widehat{\phi}_2, \hdots, \widehat{\phi}_N]^T$. The conditional distribution of $\widehat{\boldsymbol{\phi}}$ (for a given $\boldsymbol{\ell}_{P}$) is $\widehat{\boldsymbol{\phi}} \sim \mathcal{N} (\overline{\boldsymbol{\phi}}, \boldsymbol{\Omega}_{\textbf{s}})$, where $\overline{\boldsymbol{\phi}} = [\overline{\phi}_1, \overline{\phi}_2, \hdots, \overline{\phi}_N]^T$, and $\boldsymbol{\Omega}_{\textbf{s}}$ is the covariance matrix of the collections of the shadowing variables $\textbf{s} = [s_1, s_2, \hdots, s_N]^T$ given by $\{\mathbf{\Omega}_\mathbf{s}\}_{mn} = \sigma_s^2 e^{- \| \boldsymbol{\ell}_{m}-\boldsymbol{\ell}_{n} \| /X_c}$, where $X_c$ is the correlation distance within which the shadowing effects among nodes are correlated.

The DoA of the PU at the $n^{\text{th}}$ CR is given as $\theta_{n} \triangleq \arctan(\frac{y_{P} - y_{n}}{x_{P} - x_{n}}) \triangleq \angle (\boldsymbol{\ell}_{P}, \boldsymbol{\ell}_{n})$. CRs perform array signal processing techniques, such as MUSIC \cite{Schmidt1986} or ESPRIT \cite{Roy1989}, to obtain DoA estimates. The estimated DoA is commonly modeled as $\widehat{\theta}_{n} \triangleq \theta_{n}+v_{n}$ \cite{Stoica1989}, where $v_{n} \sim \mathcal{N}(0, \sigma_{n}^2)$ and $\sigma_{n}^2$ is the DoA estimation error variance. We denote $\widehat{\boldsymbol{\theta}}=[\widehat{\theta}_1, \widehat{\theta}_2, \hdots, \widehat{\theta}_N]$ as collections of DoA measurements of all CRs at the fusion center. We consider two different modelings of the DoA estimation error variance, using CRB and MUSIC algorithm, respectively. The CRB of DoA estimation error variance for unbiased DoA estimators using Uniform Linear Array (ULA) is given by \cite{Stoica1989}
\begin{equation}\label{ErrVar_CRB}
   \sigma_{n,CRB}^2 = \frac{1}{(\kappa \cos \widetilde{\theta}_n)^2} \frac{6}{N_s N_a (N_a^2 - 1) \rho_{n}},
\end{equation}
where $\kappa$ is a constant determined by the signal wavelength and array spacing, $N_s$ is the number of samples, $N_a$ is the number of antennas, $\widetilde{\theta}_n$ is the array orientation with respect to the incoming DoA defined as $\widetilde{\theta}_n \triangleq \theta_{n} - \overline{\theta}_n$, where $\overline{\theta}_n$ is the orientation of the $n^{\text{th}}$ ULA, and $\rho_{n}$ is the signal-to-noise (SNR) ratio given by $\rho_{n} = \widehat{\psi}_n/P_M$, where $P_M$ is the measurement noise power at the $n^{\text{th}}$ CR. Using the definition of SNR we can simplify (\ref{ErrVar_CRB}) as
\begin{equation}\label{ErrVar_CRB_Sim}
    \sigma_{n,CRB}^2 = \frac{6 P_M}{\kappa^2 N_s N_a (N_a^2 - 1)} \frac{1}{\widehat{\psi}_n} \frac{1}{\cos^2 \widetilde{\theta}_n} = \beta f_{CRB} (\widehat{\phi}_n) \frac{1}{\cos^2 \widetilde{\theta}_n}.
\end{equation}
where $\beta \triangleq \frac{6 P_M}{\kappa^2 N_s N_a (N_a^2 - 1)}$ and $f_{CRB} (\widehat{\phi}_n) \triangleq \frac{1}{\widehat{\psi}_n}$. The DoA estimates from MUSIC are, asymptotically in the sample size, unbiased and Gaussian distributed \cite{Stoica1989}. The estimation error variance using ULA is given by \cite{Stoica1989}
\begin{equation}\label{ErrVar_Music}
    \sigma_{n,MU}^2 = \frac{1}{(\kappa \cos \widetilde{\theta}_n)^2} \frac{6}{N_s N_a (N_a^2 - 1) \rho_{n}}(1+\frac{1}{N_a \rho_{n}}) = \beta f_{MU} (\widehat{\phi}_n) \frac{1}{\cos^2 \widetilde{\theta}_n}.
\end{equation}
where $f_{MU} (\widehat{\phi}_n) \triangleq \frac{\widehat{\psi}_n + (P_M/N_a)}{\widehat{\psi}_n^2}$. Note that both (\ref{ErrVar_CRB_Sim}) and (\ref{ErrVar_Music}) depend on RSS and PU location (used in calculation of $\widetilde{\theta}_n$).

We do not assume any specific distribution of CR placements in the derivations of Sec. \ref{sec:JCRB_Fixed}. For the asymptotic CRB derived in Sec. \ref{sec:JCRB_Asym}, we assume the CRs that can hear the PU form a circle with radius $R$ and are uniformly placed in the area. For this scenario, the distribution of $\theta_n$ is given by $\theta_n \sim \mathcal {U} [0, 2 \pi)$, and the distribution of $d_n$ is derived as
\begin{subnumcases}{p_{d_n}(r)=}
    \frac{2 r}{(R^2 - R_0^2)}, & $R_0 \leq r \leq R$ \\
    0, & otherwise,
\end{subnumcases}
where $R_0$ is a guard distance to avoid overlap between CRs and PU. The CRs are placed independently within the area, which indicates independency among all $\theta_n$'s, all $d_n$'s, and pairs of $\theta_n$ and $d_n$. We further assume the array orientation error is distributed as $\widetilde{\theta}_n = \theta_n - \overline{\theta}_n \sim \mathcal{U} (- \theta_T, \theta_T)$, which means the array orientation is uniformly distributed around the incoming DoA with parameter $\theta_T$. The PU localization problem uses RSS and DoA measurements to obtain PU location estimate $\widehat{ \boldsymbol{\ell}}_P \triangleq [\widehat{x}_p, \widehat{y}_p]^T$. A brief table of notations used in the paper is given in Table \ref{tab:Notations}.

\begin{table}
\centering
\caption{Summary of Variables Used in This Paper}
\begin{tabular}{ c | l }
    \hline
    Symbol & Meaning \\ \hline \hline
    $N$, $n$ & number of CRs, index of CRs   \\
    $R$, $R_0$ & radius of: PU coverage area, guard region  \\
    $\boldsymbol{\ell}_{n}=[x_{n}, y_{n}]^T$ & 2-dimensional coordinate of the $n^{\text{th}}$ CR  \\
    $\boldsymbol{\ell}_{P}=[x_{P}, y_{P}]^T$ & 2-dimensional coordinate of the PU  \\
    $\widehat{ \boldsymbol{\ell}}_P = [\widehat{x}_P, \widehat{y}_P]^T$ & Estimated 2-dimensional coordinate of the PU  \\
    $d_n$, $\Delta x_n$, $\Delta y_n$ & distance from the $n^{th}$ CR to the PU, distance on x-axis, distance on y-axis\\
    $P_T$, $P_M$, $\rho_n$ & PU transmit power, measurement noise power, received SNR \\
    $c_0$, $\gamma$ & channel multiplicative gain at reference distance, path-loss exponent \\
    $s_n$, $\sigma_s^2$, $X_c$& shadowing variable, its variance, correlation distance \\
    $\boldsymbol{\Omega}_{\widehat{ \boldsymbol{\ell}}_P}$, $\mathbf{\Omega}_\mathbf{s}$ & Covariance matrix of: the PU location estimates, shadowing variables \\
    $\widehat{\psi}_n$, $\widehat{\phi}_n$, $\overline{\phi}_n$& RSS measurement: in Watt, in dBm, mean in dBm \\
    $\theta_n$, $\widehat{\theta}_n$ & incoming DoA and its estimate\\
    $\overline{\theta}_n$, $\widetilde{\theta}_n$, $\theta_T$ & ULA orientation, orientation error and its distribution parameter \\
    $N_s$, $N_a$, $\kappa$ & ULA parameters: number of samples, number of antennas, array constant\\
    $\sigma_{n,CRB}^2$, $\sigma_{n,MU}^2$ & DoA error variance of the $n^{th}$ CR given by: CRB, MUSIC algorithm\\
    $\textbf{F}_{\widehat{\boldsymbol{\phi}}}$, $\textbf{F}_{\widehat{\boldsymbol{\theta}} | \widehat{\boldsymbol{\phi}}}$, $\textbf{F}$ & FIM for: RSS-only, DoA given RSS, joint RSS/DoA scenarios \\
    $\alpha$, $\beta$, $\epsilon$ & supporting constants: $\alpha \triangleq \frac{c_0 P_T e^{\sigma_s^2 /(2 \epsilon)}}{\beta}$, $\beta \triangleq \frac{6 P_M}{\kappa^2 N_s N_a (N_a^2 - 1)}$, $\epsilon \triangleq \frac{100}{(\log 10)^2}$ \\
    $\delta_0$, $\eta$& deviation distance and corresponding probability used in Theorem \ref{RSS_FIM_Theorem} and \ref{Joint_FIM_Theorem}\\
    $\textbf{s}$, $\widehat{\boldsymbol{\phi}}$, $\widehat{\boldsymbol{\theta}}$ & vector form of: shadowing variables, RSS in dBm, DoA estimates \\ \hline
    \end{tabular}
\label{tab:Notations}
\end{table}
%


\section{Joint CRBs for Fixed CR Placement}
\label{sec:JCRB_Fixed}

In this section we derive the joint CRB and the corresponding bound on RMSE for a fixed CR placement, for DoA estimates obtained from both optimal estimator and MUSIC algorithm. We also derive the CRB and RMSE for RSS-only localization as a byproduct. Using RSS and DoA as measurements, the covariance matrix of unbiased estimation of PU locations $\widehat{ \boldsymbol{\ell}}_P$ is lower-bounded by the CRB
\begin{equation}\label{LowerBound}
    \boldsymbol{\Omega}_{\widehat{ \boldsymbol{\ell}}_P} \triangleq \E \left[ \left( \widehat{ \boldsymbol{\ell}}_P - \E \left[ \widehat{ \boldsymbol{\ell}}_P \right] \right) \left( \widehat{ \boldsymbol{\ell}}_P - \E \left[ \widehat{ \boldsymbol{\ell}}_P \right] \right)^T \right] \geq \textbf{F}^{-1},
\end{equation}
where $\textbf{F}$ is the $2 \times 2$ Fisher Information Matrix (FIM) given by
\begin{equation}\label{FIM_Def}
    \textbf{F} = - \E_{\widehat{\boldsymbol{\theta}}, \widehat{\boldsymbol{\phi}}} \left[ \frac{\partial^2}{\partial \boldsymbol{\ell}_{P}^2} \log p(\widehat{\boldsymbol{\theta}}, \widehat{\boldsymbol{\phi}} | \boldsymbol{\ell}_{P}) \right].
\end{equation}
Therefore the RMSE is bounded by $RMSE \geq \sqrt{ \{ \textbf{F}^{-1} \}_{11} + \{ \textbf{F}^{-1} \}_{22} }$, where $\left\{ \textbf{X} \right\}_{ij}$ denotes the $ij^{th}$ element of matrx $\textbf{X}$. Using the standard decomposition of conditional probability $ p(\widehat{\boldsymbol{\theta}}, \widehat{\boldsymbol{\phi}} | \boldsymbol{\ell}_{P}) = p(\widehat{\boldsymbol{\theta}} | \widehat{\boldsymbol{\phi}}, \boldsymbol{\ell}_{P}) p(\widehat{\boldsymbol{\phi}} | \boldsymbol{\ell}_{P})$, the FIM is decomposed as
\begin{equation}\label{FIM_Decom}
    \textbf{F} = \left\{ - \E_{\widehat{\boldsymbol{\theta}}, \widehat{\boldsymbol{\phi}}} \left[ \frac{\partial^2}{\partial \boldsymbol{\ell}_{P}^2} \log p(\widehat{\boldsymbol{\theta}} | \widehat{\boldsymbol{\phi}}, \boldsymbol{\ell}_{P}) \right] \right\} + \left\{ - \E_{\widehat{\boldsymbol{\phi}}} \left[ \frac{\partial^2}{\partial \boldsymbol{\ell}_{P}^2} \log p(\widehat{\boldsymbol{\phi}} | \boldsymbol{\ell}_{P}) \right] \right\} \triangleq \textbf{F}_{\widehat{\boldsymbol{\theta}} | \widehat{\boldsymbol{\phi}}} + \textbf{F}_{\widehat{\boldsymbol{\phi}}}.
\end{equation}
Note that $\textbf{F}_{\widehat{\boldsymbol{\phi}}}$ is the FIM for using only RSS to localize the PU, therefore its inverse bounds the localization accuracy of algorithms that only use RSS readings. In the rest of the section, we first derive the RSS-only FIM $\textbf{F}_{\widehat{\boldsymbol{\phi}}}$. We then present the results for the joint FIM $\textbf{F}$ by deriving $\textbf{F}_{\widehat{\boldsymbol{\theta}} | \widehat{\boldsymbol{\phi}}}$ for optimal DoA estimator and MUSIC algorithm, respectively.

\subsection{RSS-only CRB}
\label{subsec:RSS_CRB}

To derive the RSS-only FIM $\textbf{F}_{\widehat{\boldsymbol{\phi}}}$, we first explicitly express the logarithm of the PDF of $\widehat{\boldsymbol{\phi}}$
\begin{equation}\label{RSS_PDF}
\log p(\widehat{\boldsymbol{\phi}} | \boldsymbol{\ell}_{P}) = - \log \left[ (2 \pi)^{N/2} (\text{det} \boldsymbol{\Omega}_{\textbf{s}})^{1/2} \right] - \frac{1}{2} (\widehat{\boldsymbol{\phi}} - \overline{\boldsymbol{\phi}})^{T} \boldsymbol{\Omega}_{\textbf{s}}^{-1} (\widehat{\boldsymbol{\phi}} - \overline{\boldsymbol{\phi}}).
\end{equation}
The RSS-only FIM $\textbf{F}_{\widehat{\boldsymbol{\phi}}}$ is then given by
\begin{equation}\label{RSS_FIM}
    \textbf{F}_{\widehat{\boldsymbol{\phi}}} = \frac{1}{2} \E_{\widehat{\boldsymbol{\phi}}} \left[ \frac{\partial^2}{\partial \boldsymbol{\ell}_{P}^2} (\widehat{\boldsymbol{\phi}} - \overline{\boldsymbol{\phi}})^{T} \boldsymbol{\Omega}_{\textbf{s}}^{-1} (\widehat{\boldsymbol{\phi}} - \overline{\boldsymbol{\phi}}) \right]
\end{equation}
The elements of $\textbf{F}_{\widehat{\boldsymbol{\phi}}}$ are derived as
\begin{eqnarray}\label{RSS_FIM_Elements}
    \left\{ \textbf{F}_{\widehat{\boldsymbol{\phi}}} \right\}_{11} &=&  \frac{\partial (\widehat{\boldsymbol{\phi}} - \overline{\boldsymbol{\phi}})^{T}}{\partial x_P}  \boldsymbol{\Omega}_{\textbf{s}}^{-1} \frac{\partial (\widehat{\boldsymbol{\phi}} - \overline{\boldsymbol{\phi}})}{\partial x_P} = \epsilon \gamma^2 \boldsymbol{\Delta} \textbf{x}^T \textbf{D}^{-2} \boldsymbol{\Omega}_{\textbf{s}}^{-1} \textbf{D}^{-2} \boldsymbol{\Delta} \textbf{x} \nonumber \\
\left\{ \textbf{F}_{\widehat{\boldsymbol{\phi}}} \right\}_{22} &=&  \frac{\partial (\widehat{\boldsymbol{\phi}} - \overline{\boldsymbol{\phi}})^{T}}{\partial y_P}  \boldsymbol{\Omega}_{\textbf{s}}^{-1} \frac{\partial (\widehat{\boldsymbol{\phi}} - \overline{\boldsymbol{\phi}})}{\partial y_P} = \epsilon \gamma^2 \boldsymbol{\Delta} \textbf{y}^T \textbf{D}^{-2} \boldsymbol{\Omega}_{\textbf{s}}^{-1} \textbf{D}^{-2} \boldsymbol{\Delta} \textbf{y} \nonumber \\
\left\{ \textbf{F}_{\widehat{\boldsymbol{\phi}}} \right\}_{12} &=& \left\{ \textbf{F}_{\widehat{\boldsymbol{\phi}}} \right\}_{21} = \frac{\partial (\widehat{\boldsymbol{\phi}} - \overline{\boldsymbol{\phi}})^{T}}{\partial x_P}  \boldsymbol{\Omega}_{\textbf{s}}^{-1} \frac{\partial (\widehat{\boldsymbol{\phi}} - \overline{\boldsymbol{\phi}})}{\partial y_P} = \epsilon \gamma^2 \boldsymbol{\Delta} \textbf{x}^T \textbf{D}^{-2} \boldsymbol{\Omega}_{\textbf{s}}^{-1} \textbf{D}^{-2} \boldsymbol{\Delta} \textbf{y},
\end{eqnarray}
where $\epsilon = 100/(\log 10)^2$, and vectors and matrices are defined as $\textbf{D} \triangleq \text{diag} (d_1, d_2, \hdots, d_N)$, $\boldsymbol{\Delta} \textbf{x} \triangleq [\Delta x_1, \Delta x_2, \hdots, \Delta x_N]^T$, $\boldsymbol{\Delta} \textbf{y} \triangleq [\Delta y_1, \Delta y_2, \hdots, \Delta y_N]^T$, and $\Delta x_n = x_P - x_n$ and $\Delta y_n = y_P - y_n$. Detailed derivations are provided in Appendix A. To obtain a compact expression of $\textbf{F}_{\widehat{\boldsymbol{\phi}}}$, let's define $\textbf{L} = [\boldsymbol{\Delta} \textbf{x}, \boldsymbol{\Delta} \textbf{y}]^T$ and $\boldsymbol{\Lambda} = \frac{1}{\epsilon \gamma^2} \textbf{D}^{2} \boldsymbol{\Omega}_{\textbf{s}} \textbf{D}^{2}$. Then it is straightforward to verify that $\textbf{F}_{\widehat{\boldsymbol{\phi}}} = \textbf{L} \boldsymbol{\Lambda}^{-1} \textbf{L}^T$. Therefore, the RMSE of RSS-only PU localization is bounded by
\begin{equation}\label{RMSE_Fix_RSS}
    \text{RMSE}_{R,F} \geq \sqrt{ \{ \textbf{F}_{\widehat{\boldsymbol{\phi}}}^{-1} \}_{11} + \{ \textbf{F}_{\widehat{\boldsymbol{\phi}}}^{-1} \}_{22} },
\end{equation}
where the subscript $R,F$ stands for RSS-only bound for fixed placement.

\subsection{Joint CRB using Optimal DoA Estimator}
\label{subsec:JCRB_CRB}

In this subsection we derive the joint CRB with DoA estimations given by the optimal estimator, using the DoA error variance given by $\sigma_{n,CRB}^2$. To derive the conditional FIM of DoA given RSS $\textbf{F}_{\widehat{\boldsymbol{\theta}} | \widehat{\boldsymbol{\phi}}}$, we first explicitly express logarithm of the conditional PDF $p(\widehat{\boldsymbol{\theta}} | \widehat{\boldsymbol{\phi}}, \boldsymbol{\ell}_{P})$ as
\begin{equation} \label{DoA_PDF}
  \log p(\widehat{\boldsymbol{\theta}} | \widehat{\boldsymbol{\phi}}, \boldsymbol{\ell}_{P}) = \sum_{n=1}^{N} \left\{ \log(\cos \widetilde{\theta}_n) - \frac{1}{2} \log \left[ 2 \pi \beta f_{CRB} (\widehat{\phi}_n) \right] - \frac{\cos^2 \widetilde{\theta}_n (\widehat{\theta}_n - \theta_n)^2}{2 \beta f_{CRB} (\widehat{\phi}_n)} \right\}.
\end{equation}
Then $\textbf{F}_{\widehat{\boldsymbol{\theta}} | \widehat{\boldsymbol{\phi}}}$ is given by
\begin{equation}\label{DoA_FIM}
    \textbf{F}_{\widehat{\boldsymbol{\theta}} | \widehat{\boldsymbol{\phi}}} = \sum_{n=1}^{N} \left\{ \E_{\widehat{\boldsymbol{\theta}}, \widehat{\boldsymbol{\phi}}} \left[ \frac{\partial^2 g_{n}}{\partial \boldsymbol{\ell}_{P}^2} \right] - \E_{\widehat{\boldsymbol{\theta}}, \widehat{\boldsymbol{\phi}}} \left[ \frac{\partial^2 h_{n}}{\partial \boldsymbol{\ell}_{P}^2} \right] \right\}
\end{equation}
where $g_{n} \triangleq \frac{\cos^2 \widetilde{\theta}_n (\widehat{\theta}_n - \theta_n)^2}{2 \beta f_{CRB} (\widehat{\phi}_n)}$ and $h_{n} \triangleq \log(\cos \widetilde{\theta}_n)$. The elements of $\textbf{F}_{\widehat{\boldsymbol{\theta}} | \widehat{\boldsymbol{\phi}}}$ are derived as
\begin{eqnarray}\label{DoA_FIM_Elements}
  \left\{ \textbf{F}_{\widehat{\boldsymbol{\theta}} | \widehat{\boldsymbol{\phi}}} \right\}_{11} &=& \sum_{n=1}^{N} \frac{\Delta y_n^2 }{d_n^4} \left\{ \frac{\alpha \cos^2 \widetilde{\theta}_n}{d_n^{\gamma}}  + 2 \tan^2 \widetilde{\theta}_n \right\} \nonumber \\
  \left\{ \textbf{F}_{\widehat{\boldsymbol{\theta}} | \widehat{\boldsymbol{\phi}}} \right\}_{22} &=& \sum_{n=1}^{N} \frac{\Delta x_n^2 }{d_n^4} \left\{ \frac{\alpha \cos^2 \widetilde{\theta}_n}{d_n^{\gamma}} + 2 \tan^2 \widetilde{\theta}_n \right\} \nonumber \\
  \left\{ \textbf{F}_{\widehat{\boldsymbol{\theta}} | \widehat{\boldsymbol{\phi}}} \right\}_{12} &=& \left\{ \textbf{F}_{\widehat{\boldsymbol{\theta}} | \widehat{\boldsymbol{\phi}}} \right\}_{21} = - \sum_{n=1}^{N} \frac{\Delta x_n \Delta y_n}{d_n^4} \left\{ \frac{\alpha \cos^2 \widetilde{\theta}_n}{d_n^{\gamma}} + 2 \tan^2 \widetilde{\theta}_n \right\}.
\end{eqnarray}
where $\alpha \triangleq c_0 P_T e^{\sigma_s^2 /(2 \epsilon)} / \beta$. Detailed derivations are provided in Appendix A. To obtain a compact expression of $\textbf{F}_{\widehat{\boldsymbol{\theta}} | \widehat{\boldsymbol{\phi}}}$, let's define $\textbf{P} = [ \boldsymbol{\Delta} \textbf{y}, - \boldsymbol{\Delta} \textbf{x} ]^T$ and $\boldsymbol{\Gamma} = \text{diag} (\gamma_{1}, \gamma_{2}, \hdots, \gamma_{N})$, where $\gamma_{n} = \frac{1}{d_n^4} \left\{ \frac{\alpha \cos^2 \widetilde{\theta}_n}{ d_n^{\gamma}} + 2 \tan^2 \widetilde{\theta}_n \right\}$. Then it is straightforward to verify that $\textbf{F}_{\widehat{\boldsymbol{\theta}} | \widehat{\boldsymbol{\phi}}} = \textbf{P} \boldsymbol{\Gamma} \textbf{P}^T$. Therefore, the joint FIM and the corresponding RMSE are given by $\textbf{F}_{J,F,C} = \textbf{P} \boldsymbol{\Gamma} \textbf{P}^T + \textbf{L} \boldsymbol{\Lambda}^{-1} \textbf{L}^T$ and
\begin{equation}
\text{RMSE}_{J,F,C} \geq \sqrt{ \{ \textbf{F}_{J,F,C}^{-1} \}_{11} + \{ \textbf{F}_{J,F,C}^{-1} \}_{22} } \label{RMSE_Fix_Joint},
\end{equation}
where the subscript $J,F,C$ represents joint CRB for fixed placement using CRB of DoA estimation error variance.

\subsection{Joint CRB using MUSIC Algorithm}
\label{subsec:JCRB_MUSIC}

In this section we derive the joint CRB with DoA estimations given by the MUSIC algorithm, using the error variance given by $\sigma_{n,MU}^2$. The conditional FIM of DoA given RSS $\textbf{F}_{\widehat{\boldsymbol{\theta}} | \widehat{\boldsymbol{\phi}}}$ is derived by replacing $f_{CRB} (\widehat{\phi}_n)$ in (\ref{DoA_PDF}) with $f_{MU} (\widehat{\phi}_n)$ given by (\ref{ErrVar_Music}). Applying the results in Appendix A, we obtain $\textbf{F}_{\widehat{\boldsymbol{\theta}} | \widehat{\boldsymbol{\phi}}} = \textbf{P} \boldsymbol{\Delta} \textbf{P}^T$, where $\boldsymbol{\Delta} = \text{diag} (\delta_{1}, \delta_{2}, \hdots, \delta_{N})$ and
\begin{equation}\label{DoA_FIM_Define_MUSIC}
  \delta_{n} = \frac{1}{d_n^4} \left\{ \frac{\cos^2 \widetilde{\theta}_n}{\beta } \left[ \frac{c_0 P_T e^{\sigma_s^2/(2 \epsilon)}}{d_n^{\gamma}} - \frac{P_M}{N_a} + \frac{P_M^2}{N_a^2} \E_{\widehat{\boldsymbol{\phi}}} \left( \frac{1}{\widehat{\psi}_n + \frac{P_M}{N_a}} \right) \right] + 2 \tan^2 \widetilde{\theta}_n \right\}.
\end{equation}
Since the RSS-only FIM is independent of DoA estimation algorithm, the joint FIM and the corresponding RMSE using MUSIC algorithm are given by $\textbf{F}_{J,F,M} = \textbf{P} \boldsymbol{\Delta} \textbf{P}^T + \textbf{L} \boldsymbol{\Lambda}^{-1} \textbf{L}^T$ and
\begin{equation}
\text{RMSE}_{J,F,M} \geq \sqrt{ \{ \textbf{F}_{J,F,M}^{-1} \}_{11} + \{ \textbf{F}_{J,F,M}^{-1} \}_{22} } \label{RMSE_MUSIC},
\end{equation}
where the subscript $J,F,M$ represents joint CRB for fixed placement using MUSIC error variance. Note that the above calculations are based on a particular fixed node placement. If the nodes are randomly placed within the area according to some distribution, say uniform distribution, integrating (\ref{RMSE_Fix_RSS}), (\ref{RMSE_Fix_Joint}) and (\ref{RMSE_MUSIC}) for $2N$ times over all node coordinates gives the theoretical average performance bound.

\section{Asymptotic Joint CRB for Uniform Random CR Placement}
\label{sec:JCRB_Asym}

The bounds on RMSE derived in Section \ref{sec:JCRB_Fixed} are useful to evaluate the achievable localization performance for a given CR placement. Numerical integration or ensemble averaging need to be performed to get the average accuracy for random CR placements. In this section we characterize the achievable localization accuracy for the most common random placement, i.e. the uniform random placement, by deriving the closed-form asymptotic CRB for such case, assuming i.i.d. shadowing and optimal DoA estimator. We also provide theoretical analysis on the required number of CRs as a function of system parameters to make the asymptotic bound tight.

\subsection{Asymptotic RSS-only CRB}
\label{subsec:Asym_RSS}

For the RSS-only FIM $\textbf{F}_{\widehat{\boldsymbol{\phi}}}$ with i.i.d. shadowing, we first rewrite it from (\ref{RSS_FIM_Elements}) as
\begin{equation}\label{RSS_FIM_Rewrite}
    \frac{1}{N} \textbf{F}_{\widehat{\boldsymbol{\phi}}} = \frac{\epsilon \gamma^2}{\sigma_s^2 N} \sum_{n=1}^{N} d_n^{-2} \left[\begin{array}{c}
                 \cos \theta_n \\
                 \sin \theta_n \\
                 \end{array}\right] [\cos \theta_n, \sin \theta_n],
\end{equation}
in which we use the fact that $[\Delta x_n, \Delta y_n] = [d_n \cos \theta_n, d_n \sin \theta_n]$. From (\ref{RSS_FIM_Rewrite}) we observe that $\frac{1}{N} \textbf{F}_{\widehat{\boldsymbol{\phi}}}$ can be interpreted as the ensemble average of a function of random variables $d_n$ and $\theta_n$, with statistical mean given by
\begin{equation}\label{RSS_FIM_Mean}
    \frac{1}{N} \E \left[\textbf{F}_{\widehat{\boldsymbol{\phi}}} \right] = \frac{\epsilon \gamma^2}{2 \sigma_s^2} \E_{d_n} [d_n^{-2}] \textbf{I}_2 = \frac{\epsilon \gamma^2 \log (R/R_0)}{\sigma_s^2 (R^2-R_0^2)} \textbf{I}_2 \triangleq f_{\widehat{\boldsymbol{\phi}}} (R, \gamma, \sigma_s^2) \textbf{I}_2,
\end{equation}
where we obtain the first equality from (\ref{RSS_FIM_Rewrite}) based on the i.i.d. distribution of $d_n$ and $\theta_n$, and $\E_{\theta_n} \left\{  [\cos \theta_n, \sin \theta_n]^T [\cos \theta_n, \sin \theta_n] \right\}=\frac{1}{2} \textbf{I}_2$, where $\textbf{I}_2$ is a $2 \times 2$ identity matrix. The deviation probability of the ensemble average $\frac{1}{N} \textbf{F}_{\widehat{\boldsymbol{\phi}}}$ from the statistical mean $\frac{1}{N} \E \left[\textbf{F}_{\widehat{\boldsymbol{\phi}}} \right]$ is given by the following theorem.

\begin{theorem}
\emph{(Deviation Probability of RSS-only FIM)}
\label{RSS_FIM_Theorem}
For any $\delta_0 > 0$, we have
\begin{equation}\label{Theorem1_Eqn1}
    \text{Pr} \left\{ \left\| \frac{1}{N} \textbf{F}_{\widehat{\boldsymbol{\phi}}} - \frac{1}{N} \E \left[ \textbf{F}_{\widehat{\boldsymbol{\phi}}} \right] \right\|_F > \delta_0 \right\} < \frac{2 \epsilon^2 \gamma^4}{\sigma_s^4 \delta_0^2 N} \left[ \frac{1}{2 R^2 R_0^2} - \frac{\log^2(R/R_0)}{(R^2-R_0^2)^2} \right],
\end{equation}
where $\left\| \cdot \right\|$ represents Frobenius matrix norm.
\end{theorem}

The proof of the theorem is based on Chebyshev's inequality for random matrix, and all details are provided in Appendix B. The intuition from Theorem \ref{RSS_FIM_Theorem} is that we can well-approximate $\frac{1}{N} \textbf{F}_{\widehat{\boldsymbol{\phi}}}$ with $\frac{1}{N} \E \left[\textbf{F}_{\widehat{\boldsymbol{\phi}}} \right]$ when $N$ is large enough. In order for the deviation probability to be smaller than a pre-defined threshold $\eta \in (0,1)$, we need to bound the right-hand-side of (\ref{Theorem1_Eqn1}) by $\eta$. The required number of CRs is then bounded by
\begin{equation}\label{Required_N_RSS}
    N \geq \frac{2 \epsilon^2 \gamma^4}{\sigma_s^4 \delta_0^2 \eta} \left[ \frac{1}{2 R^2 R_0^2} - \frac{\log^2(R/R_0)}{(R^2-R_0^2)^2} \right].
\end{equation}
Applying the approximation, the RMSE of RSS-only algorithms for uniform random CR placement is given by
\begin{equation}\label{RMSE_Uniform_RSS}
\text{RMSE}_{R,U} \geq \left(\frac{2}{N f_{\widehat{\boldsymbol{\phi}}} (R, \gamma, \sigma_s^2) } \right)^{1/2},
\end{equation}
where the subscript $R,U$ represents RSS-only CRB for uniform placement.

\subsection{Asymptotic Joint CRB}
\label{subsec:Asym_Joint}

We apply similar procedure in Sec. \ref{subsec:Asym_RSS} to derive the asymptotic joint CRB from $\textbf{F} = \textbf{F}_{\widehat{\boldsymbol{\phi}}} + \textbf{F}_{\widehat{\boldsymbol{\theta}} | \widehat{\boldsymbol{\phi}}}$. We first rewrite $\textbf{F}_{\widehat{\boldsymbol{\theta}} | \widehat{\boldsymbol{\phi}}}$ from (\ref{DoA_FIM_Elements}) as
\begin{equation}\label{DoA_FIM_Rewrite}
    \frac{1}{N} \textbf{F}_{\widehat{\boldsymbol{\theta}} | \widehat{\boldsymbol{\phi}}} = \frac{1}{N} \sum_{n=1}^{N} \left[ \alpha d_n^{-(\gamma+2)} \cos^2 \widetilde{\theta}_n + 2 d_n^{-2} \tan^2 \widetilde{\theta}_n \right] \left[\begin{array}{c}
                 \sin \theta_n \\
                 - \cos \theta_n \\
                 \end{array}\right] [\sin \theta_n, - \cos \theta_n].
\end{equation}
From (\ref{DoA_FIM_Rewrite}) we observe that $\frac{1}{N} \textbf{F}_{\widehat{\boldsymbol{\theta}} | \widehat{\boldsymbol{\phi}}}$ can be interpreted as the ensemble average of a function of random variables $d_n$, $\theta_n$ and $\widetilde{\theta}_n$, with statistical mean given by
\begin{eqnarray}\label{DoA_FIM_Mean}
    \frac{1}{N} \E \left[ \textbf{F}_{\widehat{\boldsymbol{\theta}} | \widehat{\boldsymbol{\phi}}} \right] &=& \frac{1}{(R^2-R_0^2)} \left\{ 2 \log (R/R_0) \left(\frac{\tan \theta_T}{\theta_T} - 1\right) - \frac{\alpha}{2 \gamma} (R^{- \gamma} - R_0^{- \gamma}) \left( \frac{\sin 2 \theta_T}{2 \theta_T} +1\right) \right\} \textbf{I}_2,\nonumber \\
&=& f_{\widehat{\boldsymbol{\theta}} | \widehat{\boldsymbol{\phi}}} (R, \gamma, \sigma_s^2, \theta_T, \beta) \textbf{I}_2.
\end{eqnarray}
The deviation probability of the ensemble average $\frac{1}{N} \textbf{F}$ from the statistical mean $\frac{1}{N} \E \left[ \textbf{F} \right] = \frac{1}{N} \E \left[\textbf{F}_{\widehat{\boldsymbol{\phi}}} \right] + \frac{1}{N} \E \left[ \textbf{F}_{\widehat{\boldsymbol{\theta}} | \widehat{\boldsymbol{\phi}}} \right]$ is given by the following theorem.

\begin{theorem}
\emph{(Deviation Probability of Joint FIM)}
\label{Joint_FIM_Theorem}
For any $\delta_0 > 0$, we have
\begin{equation}\label{Theorem2_Eqn1}
    \text{Pr} \left\{ \left\| \frac{1}{N} \textbf{F} - \frac{1}{N} \E \left[ \textbf{F} \right] \right\|_F > \delta_0 \right\} < \frac{1}{N \delta_0^2} \left\{ \frac{\epsilon^2 \gamma^4}{\sigma_s^4 R^2 R_0^2}  + \E \left[ f_n^2 \right] - \frac{1}{2} \left[ \E \left[ f_n \right] + \frac{2 \epsilon \gamma^2 \log(R/R_0)}{\sigma_s^2 (R^2 - R_0^2)}\right]^2  \right\}.
\end{equation}
where $f_n \triangleq \alpha d_n^{-(\gamma+2)} \cos^2 \widetilde{\theta}_n + 2 d_n^{-2} \tan^2 \widetilde{\theta}_n$, $\E \left[ f_n \right] = 2 f_{\widehat{\boldsymbol{\theta}} | \widehat{\boldsymbol{\phi}}} (R, \gamma, \sigma_s^2, \theta_T, \beta)$ and
\begin{eqnarray}\label{Theorem2_Eqn2}
  \E \left[ f_n^2 \right] &=& \frac{4 \tan^5 \theta_T}{5 \theta_T R^2 R_0^2} \times {_2F_1} \left( \frac{5}{2}, 1, \frac{7}{2}, -\tan^2 \theta_T\right) - \frac{2 \alpha \left( R^{-(\gamma+2)} - R_0^{-(\gamma+2)}\right)}{(\gamma+2)(R^2-R_0^2)} \left( 2 - \frac{\sin 2 \theta_T}{\theta_T} \right) \nonumber \\
&& - \frac{\alpha^2 \left( \sin 4 \theta_T + 8 \sin 2 \theta_T + 36 \theta_T \right) \left[ R^{-2(\gamma+1)} - R_0^{-2(\gamma+1)} \right]}{32 \theta_T (\gamma+1) (R^2 -R_0^2)}
\end{eqnarray}
and ${_2F_1}(x)$ is the hypergeometric function.
\end{theorem}

The proof of the theorem is provided in Appendix B. The intuition from Theorem \ref{Joint_FIM_Theorem} is that we can well-approximate $\frac{1}{N} \textbf{F}$ with $\frac{1}{N} \E \left[\textbf{F}\right]$ when $N$ is large enough. In order for the deviation probability to be smaller than $\eta$, we need to bound the right-hand-side of (\ref{Theorem2_Eqn1}) by $\eta$. The required number of CRs is then bounded by
\begin{equation}\label{Required_N_Joint}
    N \geq \frac{1}{\eta \delta_0^2} \left\{ \frac{\epsilon^2 \gamma^4}{\sigma_s^4 R^2 R_0^2}  + \E \left[ f_n^2 \right] - \frac{1}{2} \left[ \E \left[ f_n \right] + \frac{2 \epsilon \gamma^2 \log(R/R_0)}{\sigma_s^2 (R^2 - R_0^2)}\right]^2  \right\}.
\end{equation}
Applying the approximation, the joint FIM and the corresponding RMSE are given by $\textbf{F}_{J,U,C} = N \left[ f_{\widehat{\boldsymbol{\phi}}} (R, \gamma, \sigma_s^2) + f_{\widehat{\boldsymbol{\theta}} | \widehat{\boldsymbol{\phi}}} (R, \gamma, \sigma_s^2, \theta_T, \beta) \right] \textbf{I}_2$ and
\begin{equation}\label{RMSE_Uniform_Joint}
     \text{RMSE}_{J,U,C} \geq \left\{ \frac{2}{N \left[ f_{\widehat{\boldsymbol{\phi}}} (R, \gamma, \sigma_s^2) + f_{\widehat{\boldsymbol{\theta}} | \widehat{\boldsymbol{\phi}}} (R, \gamma, \sigma_s^2, \theta_T, \beta) \right]} \right\}^{1/2},
\end{equation}
where the subscript $J,U,C$ represents joint CRB for uniform placement using CRB of DoA estimation error variance. Note that different from the RMSE expression (\ref{RMSE_Fix_Joint}) which is conditioned on a fixed CR placement, the result in (\ref{RMSE_Uniform_Joint}) provides asymptotic closed-form expression of RMSE for uniform CR placement. Since the Chebyshev's inequality is not a very tight bound on the derivation probability, the results we obtained may overestimate the required number of CRs. Although similar analysis can be performed based on more accurate mathematical tools, such as central limit theorem (approximate the random matrix as Gaussian distributed) or Chernoff bound (requires the tail of the random variable distribution to be Gaussian-like) \cite{Proakis2001}, the results obtained by the Chebyshev's inequality are more general in that they do not assume any distribution of the random matrix. Therefore, (\ref{Required_N_RSS}) and (\ref{Required_N_Joint}) mainly serve to provide insights of impact of system parameters on the required number of CRs to make the asymptotic bounds accurate.


\section{Numerical Results}
\label{sec:Simulation}

In this section we present numerical results to verify the derived CRBs and study the achievable localization performance for different node density, channel conditions and array parameters. We compare the derived bounds with the joint CRB with constant RSS derived in Eqn (6) of \cite{Penna2011} to show the theoretical advantage of the new bound. Performance of practical localization algorithms is also included to investigate how close the CRBs can be achieved. A numerical study on the accuracy of the asymptotic CRBs derived in Section IV is presented at the end of the section, highlighting the conditions when the asymptotic CRBs tightly approximate the exact bounds. We start by a brief introduction of the considered localization algorithms.

\subsection{Practical Localization Algorithms}

We have selected two cooperative localization algorithms, one for RSS-only localization and the other one for joint RSS/DoA localization. For the RSS-only case, Weighted Centroid Localization (WCL) is a range-free algorithm that provides low-complexity and coarse-grained location estimates. In this technique, PU location is approximated as the weighted average of all secondary user positions within its transmission range, where the weights are proportional to RSS of each CR. The WCL location estimate is given by $\widehat{ \boldsymbol{\ell}}_{P,WCL} = \sum_{n=1}^{N} \widehat{\psi}_n \boldsymbol{\ell}_n / \sum_{n=1}^{N} \widehat{\psi}_n$. WCL does not require knowledge of PU transmit power and channel conditions, and is robust to shadowing variance, compared to range-based algorithms such as Lateration \cite{Wang2011}.

For the joint RSS/DoA case, weighted Stansfield algorithm is shown to outperform other DoA fusion algorithms, such as maximum likelihood algorithm solved by iteration methods and linear least square algorithm \cite{Wang2012}. The weights for each DoA are their estimated error variances, which are function of RSS. The weighted Stansfield location estimate is given by $\widehat{ \boldsymbol{\ell}}_{P,St} = (\textbf{A}_{St}^T \textbf{W}^{-1} \textbf{A}_{St})^{-1} \textbf{A}_{St}^T \textbf{W}^{-1} \textbf{b}_{St}$, where
\begin{equation}\label{Stan_Def}
  \textbf{A}_{St} = \left[\begin{array}{cc}
                           \sin(\widehat{\theta}_1) & - \cos(\widehat{\theta}_1) \\
                           \vdots      & \vdots \\
                           \sin(\widehat{\theta}_N) & - \cos(\widehat{\theta}_N) \\
                           \end{array}\right],
\textbf{b}_{St} = \left[\begin{array}{c}
                           x_1 \sin(\widehat{\theta}_1) - y_1 \cos(\widehat{\theta}_1) \\
                           \vdots \\
                           x_1 \sin(\widehat{\theta}_N) - y_1 \cos(\widehat{\theta}_N) \\
                           \end{array}\right]
\end{equation}
and the weighting matrix is defined as $\textbf{W} = \text{diag} [\widehat{\sigma}_1^2, \hdots, \widehat{\sigma}_N^2]$, where $\widehat{\sigma}_n^2$ is obtained by replacing true DoA $\theta_n$ with estimated DoA $\widehat{\theta}_n$ in (\ref{ErrVar_CRB_Sim}) or (\ref{ErrVar_Music}).

\subsection{Simulation Settings}

We summarize the basic parameter settings in our simulations as follows. The CRs are uniformly placed in a circle of radius $R=150$m and protective region $R_0=5$m, centered around the PU with location $\boldsymbol{\ell}_{P}=[0, 0]^T$ for simplicity and without loss of generality. The PU transmit power $P_T$ is $20$dBm ($100$mW) which is the typical radiation power of IEEE 802.11b/g wireless LAN transmitters for $20$MHz channels in ISM bands. The measurement noise power $P_M$ is $-80$dBm ($10$pW) which is a moderate estimate of noise introduced in the DoA estimation process, compare to the $-100$dBm ($0.1$pW) thermal noise power of the 802.11 WLAN channel. The path-loss exponent and shadowing variance are $\gamma=5$ and $\sigma_s = 6$dB respectively. The power and channel settings result in an averaged and minimum received SNR of $10$dB and $-10$dB, respectively. Each CR is equipped with an ULA with $N_a = 2$ antennas, and uses $N_s=50$ samples for each localization period. The array orientation error with regard to the incoming DoA is assumed to follow uniform distribution given by $\widehat{\theta}_n \sim \mathcal{U} (- \pi/3, \pi/3)$. We use the parameter $\pi/3$ for the orientation error since the maximum orientation error is $\pi/2$ which gives infinite error variance as can be verified through (\ref{ErrVar_CRB_Sim}) or (\ref{ErrVar_Music}), thus $\pi/3$ is a practical value. Each data point in the figures is obtained from averages of 1000 CR placements (with random array orientation) and 2000 realization of RSS/DoA measurements if applicable. In the following subsections, we use these settings unless stated otherwise.

\subsection{Impact of Number of CRs}

We first study the impact of node density on the localization accuracy. The results for the RSS-only CRB (\ref{RMSE_Fix_RSS}) and WCL algorithm are presented in Fig.\ref{fig:VaryNode_RSS}. It is observed that adding more CRs provides steady performance improvement for both RSS-only CRB and WCL algorithm, and the WCL algorithm has a constant RMSE gap of about $10$m for sufficiently large number of CRs, say 40, compared with the RSS-only CRB. The results for the derived joint CRB using optimal DoA estimators (\ref{RMSE_Fix_Joint}) and the weighted Stansfield algorithm are presented in Fig.\ref{fig:VaryNode_Joint}. The localization accuracy of the joint CRB and algorithm vary from 0.1 meters to 5 meters, compared with the accuracy range of 5 to 35 meters for RSS-only CRB and WCL in Fig.\ref{fig:VaryNode_RSS}, indicating that using both RSS and DoA measurements provides much better accuracy than using only RSS. The predicted accuracy by the joint CRB is achievable by the weighted Stansfiled algorithm, for large number of CRs.

\begin{figure}
\centering{
\includegraphics[width=0.6\columnwidth]{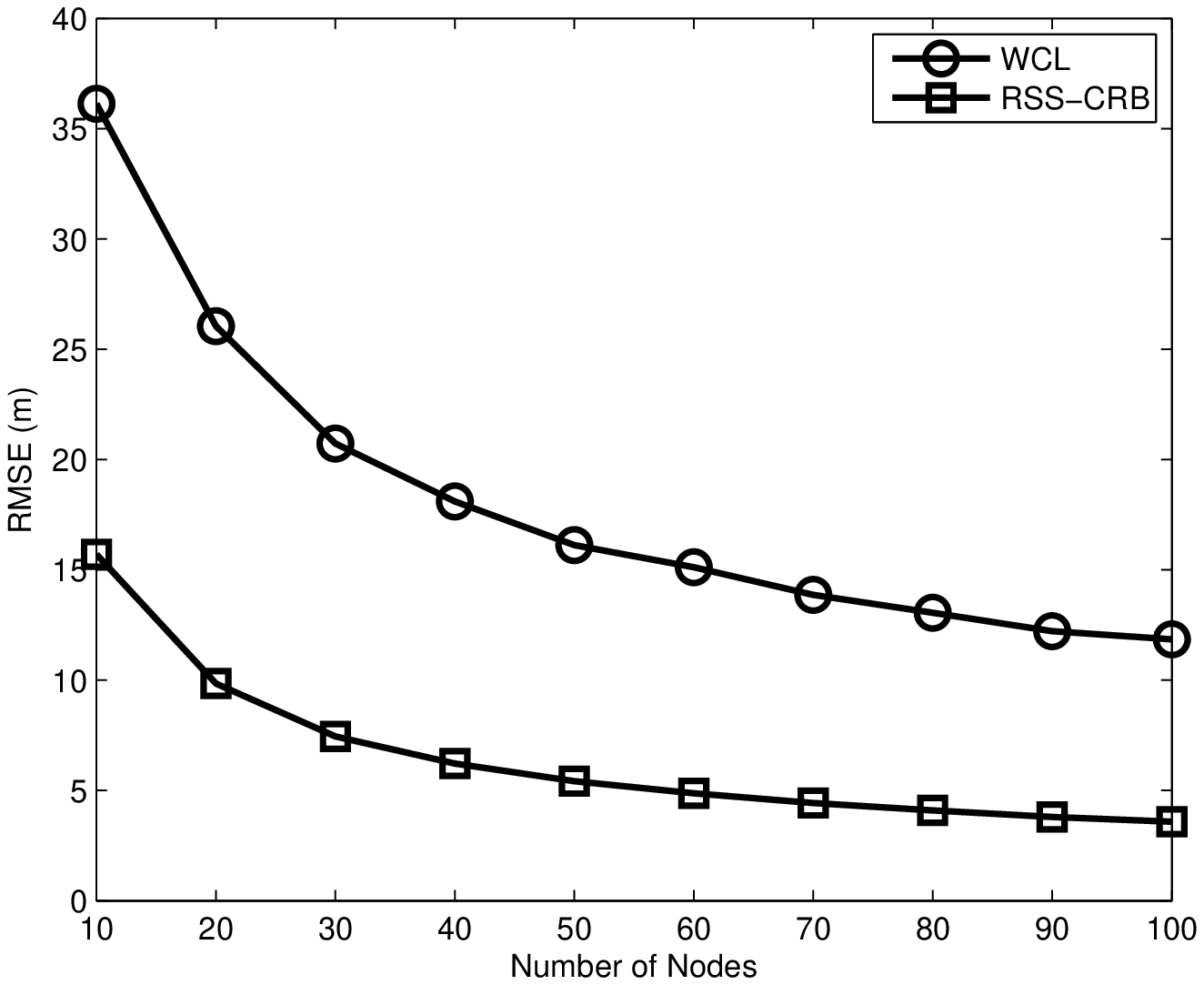}}
\caption{Results for RMSE of RSS-only CRB and WCL algorithm with varying number of CRs, with uniform random placement in a circle with $R=150$m.}
\label{fig:VaryNode_RSS}
\end{figure}
\begin{figure}
\centering{
\includegraphics[width=0.6\columnwidth]{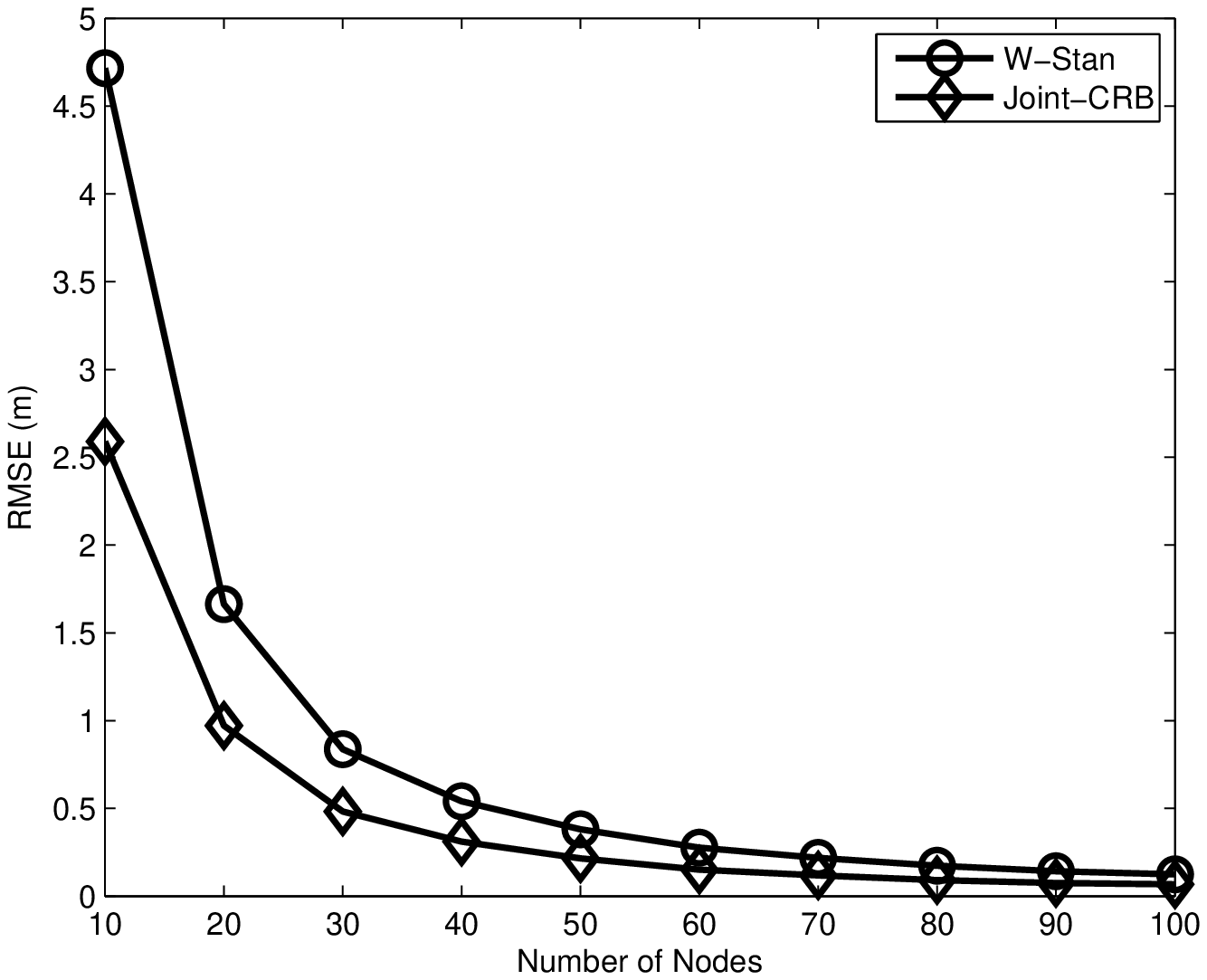}}
\caption{Results for RMSE of the joint CRB and weighted Stansfield algorithm with varying number of CRs, with uniform random placement in a circle with $R=150$m.}
\label{fig:VaryNode_Joint}
\end{figure}

\subsection{Impact of Channel Conditions}

We study the impact of channel conditions including shadowing variance and correlation distance on the localization accuracy. The considered values of shadowing variance vary from 4 to 10dB, and the typical correlation distance for mobile communication is shown to be within $30$ meters \cite{Taaghol1997}. The results for the RSS-only CRB and the WCL algorithms are shown in Fig.\ref{fig:VarySha_RSS}. It is observed that increasing shadowing standard derivation causes a linear increase in the RMSE, in that a 2dB increase of shadowing standard derivation results in 4 meters accuracy loss for both RSS-only CRB and the WCL algorithm. The performance loss caused by correlated shadowing is more significant for the RSS-only CRB than WCL. For example, when the correlation distance increases from 0 meter to 30 meters under 10dB shadowing, the RMSE rises 5 meters for the RSS-only CRB, however the loss is about 2 meters for WCL. Therefore, WCL is robust to correlated shadowing.

\begin{figure}
\centering{
\includegraphics[width=0.6\columnwidth]{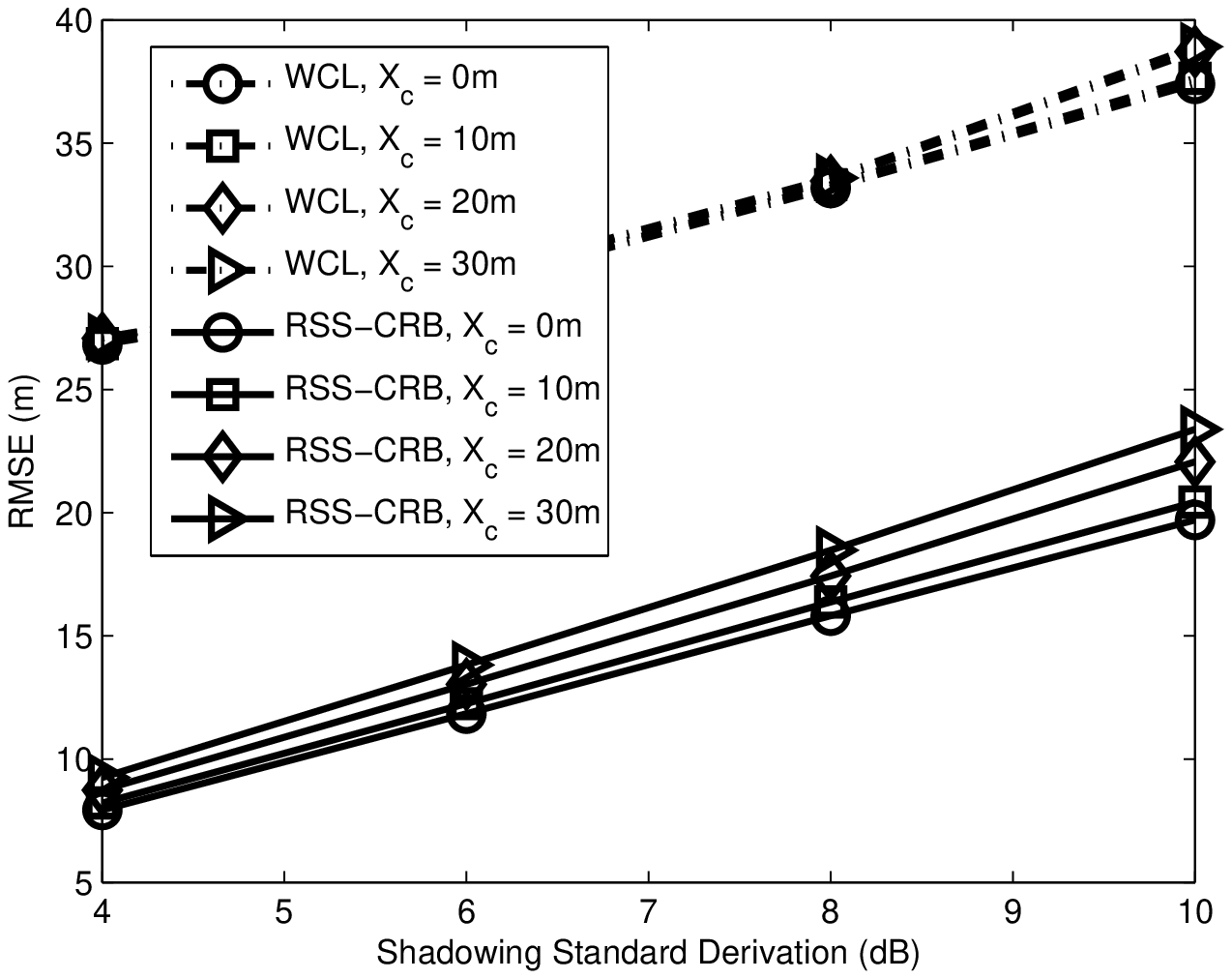}}
\caption{Results for RMSE of RSS-only CRB and WCL algorithm with varying shadowing standard derivation and correlation distance, with uniform random placement of 15 CRs.}
\label{fig:VarySha_RSS}
\end{figure}

The impact of shadowing standard derivation on the accuracy of derived joint CRB, the joint CRB of \cite{Penna2011} and the weighted Stansfield algorithm are presented in Fig.\ref{fig:VarySha_Joint}. Both results using optimal DoA estimator and the MUSIC algorithm are shown. We obtain a very surprising observation that when the shadowing standard derivation increases, the RMSEs of all considered cases decrease. This is because the RSS $\widehat{\psi}_n$ is a log-normal random variable, when the shadowing standard derivation increases, the bell-shape PDF of the log-normal variable will become more spread-out, resulting in higher probabilities for larger RSS and lower RSS observations. Since the DoA estimation error variances (\ref{ErrVar_CRB_Sim}) and (\ref{ErrVar_Music}) are inversely proportional to RSS, increasing shadowing results in more CRs with both good and bad DoA measurements. Therefore with proper weighting based on RSS information in the DoA fusion process, we may obtain better accuracy for large shadowing cases. The weighted Stansfield algorithm cannot achieve the joint CRB since it uses the estimated error variance for the weighting scheme; compared to the perfect knowledge of error variance obtained by the joint CRB. We can also observe that using MUSIC algorithm causes a 1 meter RMSE loss for weighted Stansfield algorithm, which is more significant compared to the joint CRB. Using MUSIC algorithm has small impact on the joint CRB, since the increase in DoA estimation error is mitigated by making use of the perfect knowledge of each DoA's quality. Comparing the derived joint CRB and the CRB of \cite{Penna2011} we conclude that the new CRB further reveals the potential achievable accuracy of the joint localization algorithms, therefore is more informative.

The impact of correlation distance on the accuracy of derived joint CRB and the weighted Stansfield algorithm are presented in Fig.\ref{fig:VaryCorr_Joint}. We observe that both the joint CRB and the weighted Stansfield algorithm are robust to correlated shadowing, in that the performance for joint CRB is almost constant for all $X_c$ values, and the RMSE loss for weighted Stansfiled algorithm is only $0.2$ meter from i.i.d shadowing to $X_c=30$m. The robustness of the joint CRB and weighted Stansfield algorithm to correlated shadowing is due to the fact that correlation only affect the RSS measurements which determine the variance of DoA measurements; however the DoAs at each CR are corrupted by independent noise. Therefore the correlated shadowing does not have significant impact on the DoA fusion, and the resulting RMSE.

\begin{figure}
\centering{
\includegraphics[width=0.6\columnwidth]{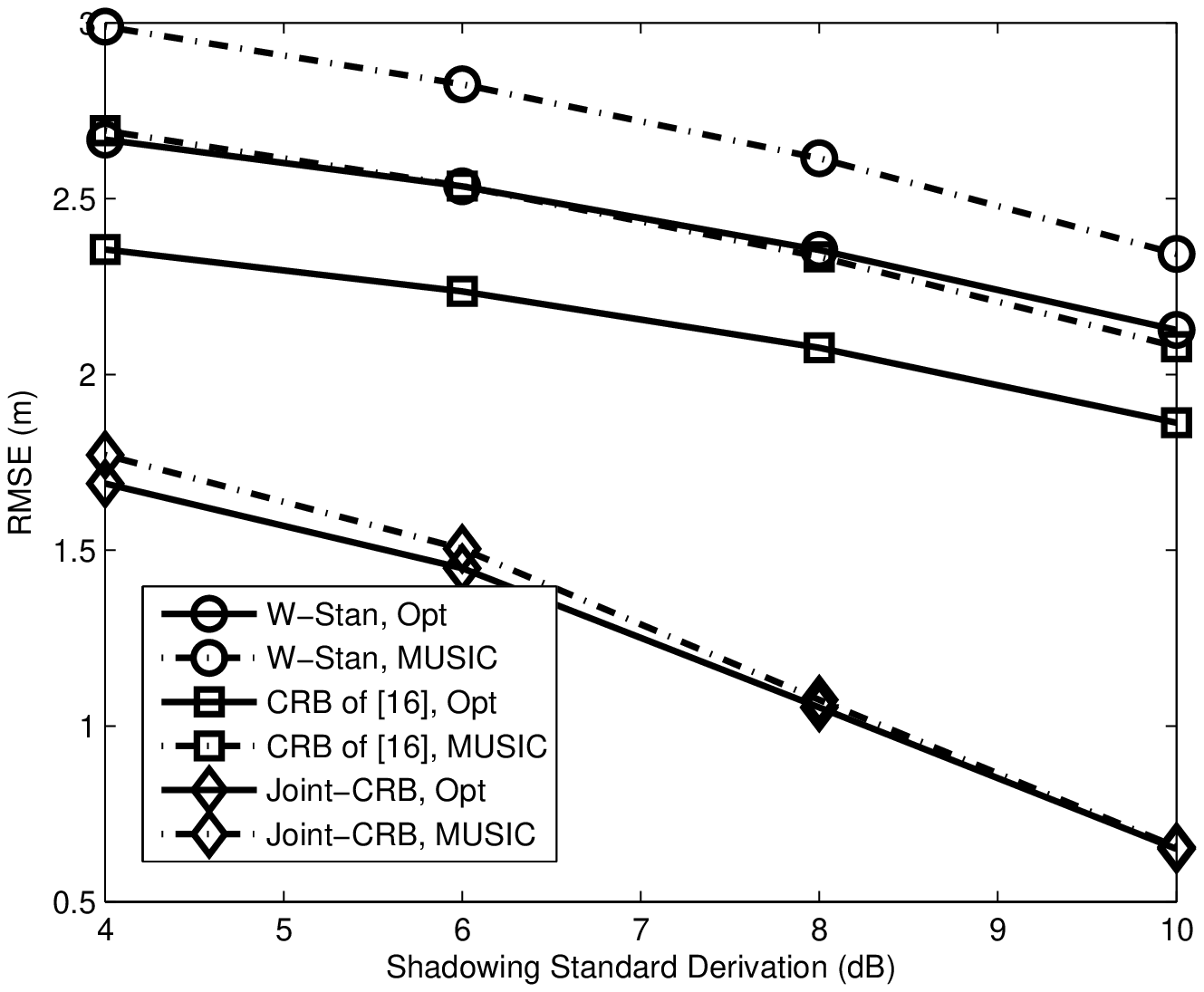}}
\caption{Results for RMSE of the joint CRB, CRB of \cite{Penna2011} and weighted Stansfield algorithm with varying shadowing standard derivation, with uniform random placement 15 CRs.}
\label{fig:VarySha_Joint}
\end{figure}
\begin{figure}
\centering{
\includegraphics[width=0.6\columnwidth]{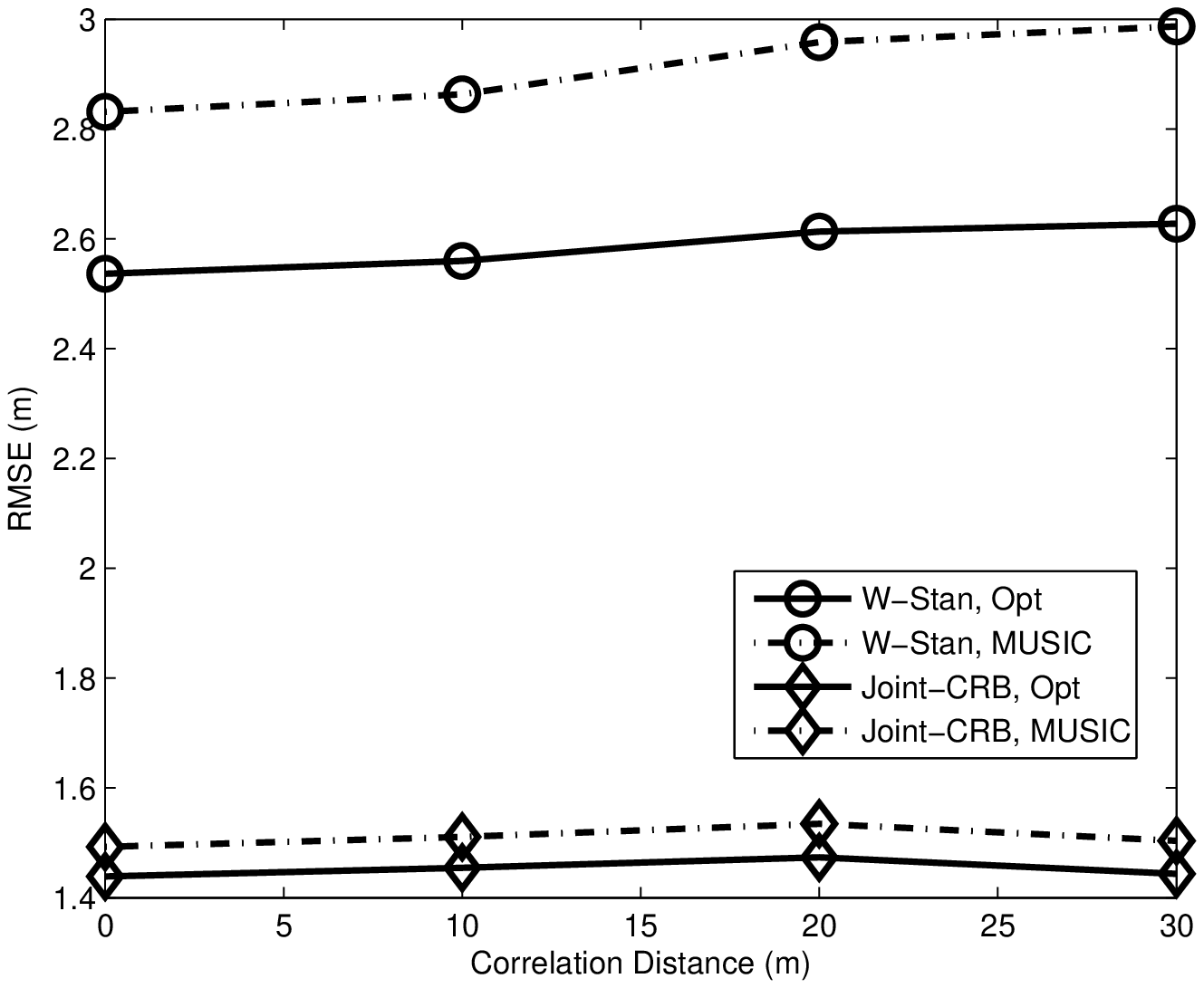}}
\caption{Results for RMSE of the joint CRB and weighted Stansfield algorithm with varying correlation distance, with uniform random placement 15 CRs.}
\label{fig:VaryCorr_Joint}
\end{figure}

\subsection{Impact of Array Parameters}

In this subsection we study the impact of various array parameters, including number of antennas, number of samples and array orientation error, on the localization performance. The comparison results for varying number of antennas are shown in Fig.\ref{fig:VaryNa}. It is observed that the benefit of adding number of antennas from 2 to 3 is about twice compared with 3 to 4; and the benefit continues to saturate for adding even more antennas. This indicates that in practical systems, using a small number of antennas at each node is good from both hardware cost and localization accuracy point of view.

\begin{figure}
\centering{
\includegraphics[width=0.6\columnwidth]{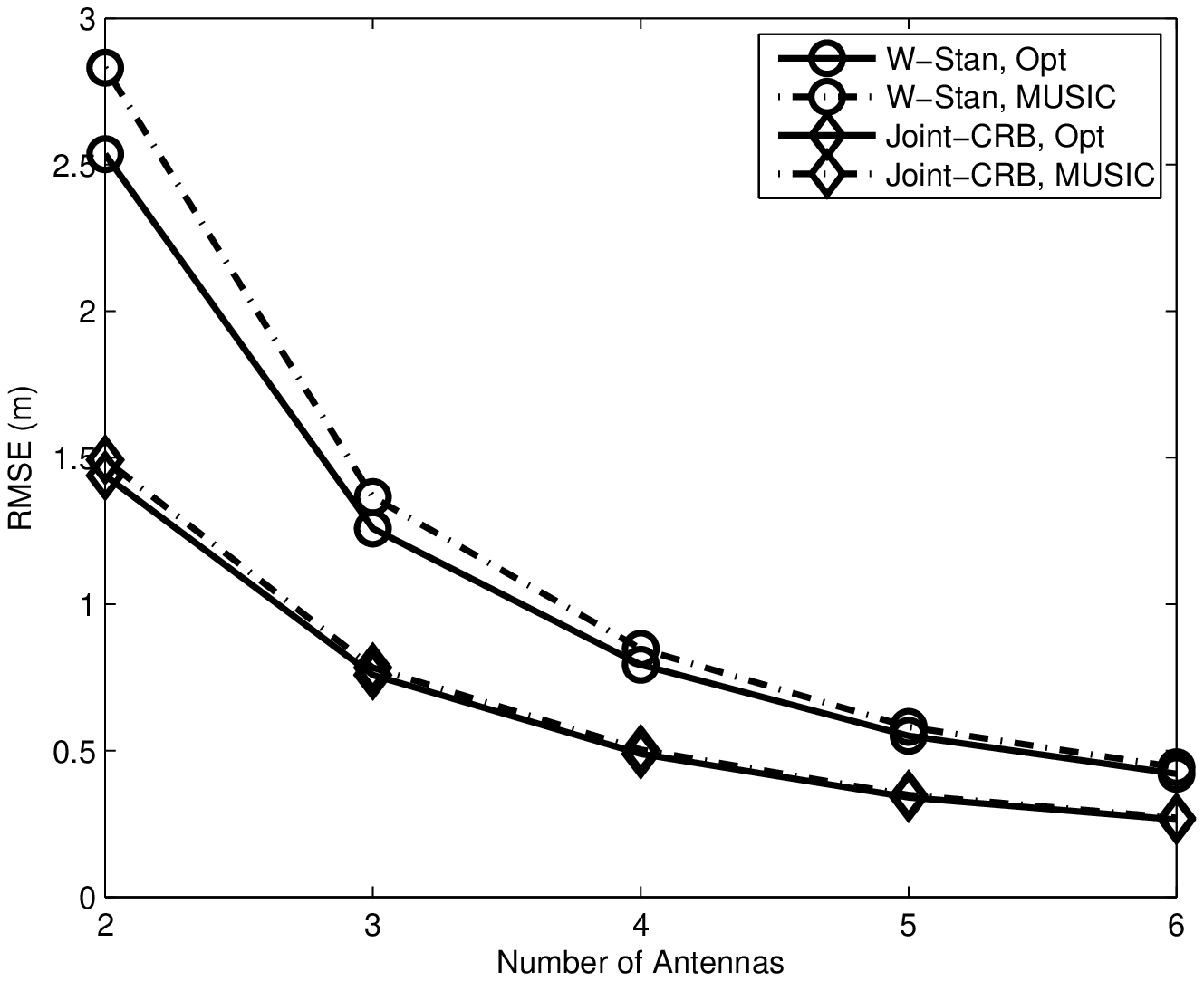}}
\caption{Results for RMSE of the joint CRB and weighted Stansfield algorithm with varying number of antennas, with uniform random placement 15 CRs. For this scenario, $N_s =50$.}
\label{fig:VaryNa}
\end{figure}

The number of samples used in practical systems is affected by several factors such as signal bandwidth, sampling period, spectrum sensing regulations etc, which are beyond the scope of this paper. Therefore we provide results for some typical values of $N_s$ to gain some insights on the impact of this parameter on the RMSE, with results shown in Fig.\ref{fig:VaryNs}. It is observed that the benefit of adding samples also tends to saturate when the number of samples is large enough, say $100$. Comparing the results with Fig.\ref{fig:VaryNa} we conclude that increasing number of antennas introduces more significant improvement on RMSE. For example, a combination of 3 antennas and 50 samples achieves RMSE of $0.75$ meters; however this accuracy is not achievable for 2 antennas, even through we reach 150 samples.

\begin{figure}
\centering{
\includegraphics[width=0.6\columnwidth]{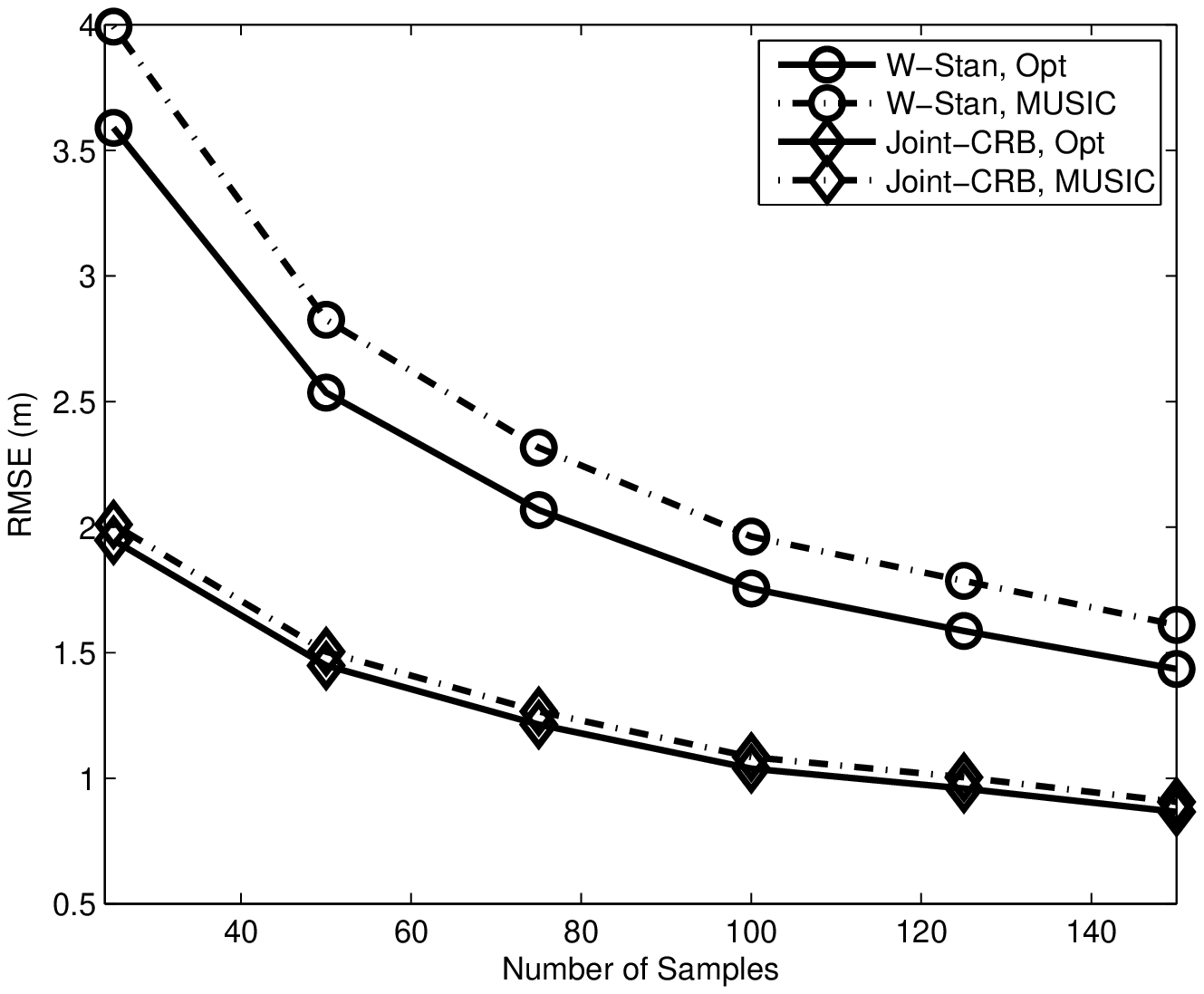}}
\caption{Results for RMSE of the joint CRB and weighted Stansfield algorithm with varying number of samples, with uniform random placement 15 CRs. For this scenario, $N_a=2$.}
\label{fig:VaryNs}
\end{figure}

The impact of array orientation error parameter $\theta_T$ is studied by results shown in Fig.\ref{fig:VaryOrt}. It is observed that the joint CRB is robust to $\theta_T$ as the RMSE increases only 0.3 meters from $\theta_T=0$ to $\theta_T = 5\pi/12$. However the weighted Stansfield algorithm can only tolerate small orientations errors, and the RMSE rises by 1 meter from $\theta_T = \pi/4$ to $\theta_T = 5\pi/12$ due to difficulties to accurately estimate DoA estimation error variances for large orientation errors.

\begin{figure}
\centering{
\includegraphics[width=0.6\columnwidth]{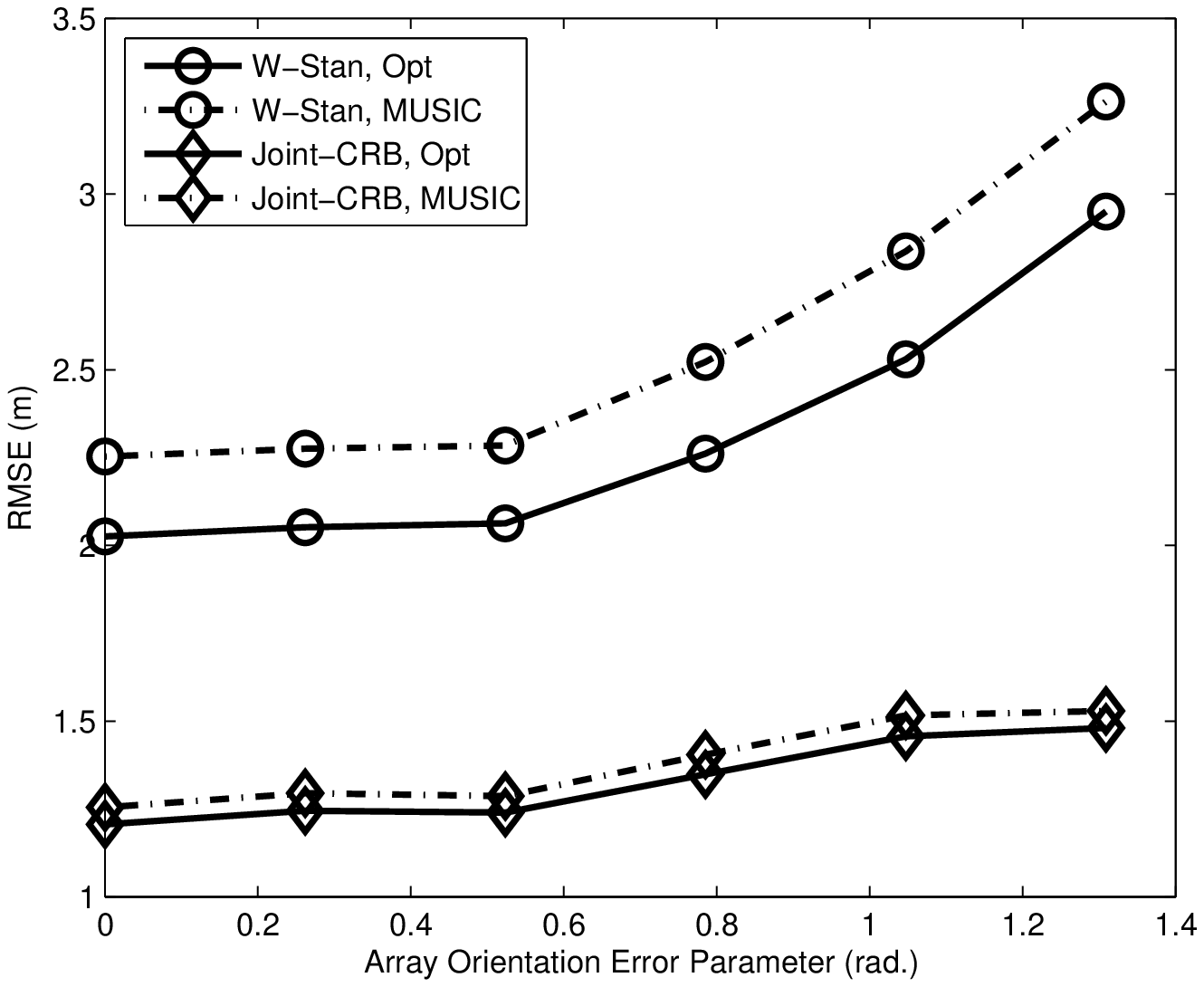}}
\caption{Results for RMSE of the joint CRB and weighted Stansfield algorithm with varying array orientation error parameter, with uniform random placement 15 CRs.}
\label{fig:VaryOrt}
\end{figure}

\subsection{Accuracy of Asymptotic CRBs}

We now evaluate the accuracy of the asymptotic CRBs for uniform random placement derived in Section IV compared with the exact CRBs conditioned on a specific placement derived in Section III. From (\ref{Required_N_RSS}) and (\ref{Required_N_Joint}) we observe that the required number of CRs to make the asymptotic CRB accurate is a function of various array and channel parameters. It is straightforward to numerically investigate impact of each system parameter on the required $N$. Among all involved parameters, we notice that the value of $R_0/R$ plays an important role in determining the required $N$. Thus we provide the comparison of exact CRBs after numerically averaging over uniform CR placements with the asymptotic CRBs for different $R_0$, with results shown in Fig.\ref{fig:Asym_RSS} and Fig.\ref{fig:Asym_Joint} for RSS-only bounds and joint bounds, respectively. For RSS-only bounds, the asymptotic CRB well approaches the exact bounds, in that the difference is less than 1 meter for $N=40$ when $R_0/R=0.1$, and the accuracy of the asymptotic bounds is almost the same as the exact bounds for $R_0/R=0.3$ and $R_0/R=0.5$. Note that the RMSE increases as we increase $R_0$ since we are excluding the nodes that are closer to the PU. On the contrary, since the variance of the asymptotic joint bound is larger than the asymptotic RSS-only bound, the approximation is less accurate. As we can observe from the figure, the approximation is satisfactory for $N>40$ when $R_0/R=0.3$m, and $N>30$ when $R_0/R=0.5$m. In practical systems, the circle with radius $R_0$ correspond to the 'forbidden region' around the PU where activities of secondary network will severely interfere the PU and be interfered by the PU, the value of $R_0$ depends on primary network requirements. However, if we consider the area ratio given by $R_0^2/R^2$, then both the asymptotic RSS-only and joint CRB provide satisfying accuracy for area ratio greater than $9 \%$, which is a very reasonable forbidden region setting.

We also evaluated the required number of CRs for asymptotic RSS-only and joint CRBs, predicted by (\ref{Required_N_RSS}) and (\ref{Required_N_Joint}) respectively, for varying $R_0/R$. The results are shown in Fig. \ref{fig:RequiredN}. Since it is not possible to directly relate the amount of deviation $\delta_0$ in Theorem \ref{RSS_FIM_Theorem} and \ref{Joint_FIM_Theorem} with the gap between the asymptotic CRBs and the exact CRBs, the choices of $\delta_0$ and $\eta$ become design parameters. As we have explained in Sec. \ref{sec:JCRB_Asym}, (\ref{Required_N_RSS}) and (\ref{Required_N_Joint}) tend to overestimate the required $N$, therefore we cannot select very small $\delta_0$ and $\eta$. In Fig. \ref{fig:RequiredN}, $\delta_0 = \left\| \frac{1}{N} \E \left[ \textbf{F}_{\widehat{\boldsymbol{\phi}}} \right] \right\|_F$ for RSS-only case, $\delta_0 = 2 \left\| \frac{1}{N} \E \left[ \textbf{F} \right] \right\|_F$ for the joint case and $\eta = 0.15$, which is the probability that the amount of deviation of the ensemble average is comparable to the statistical mean is 15\%. We use a larger constant for the joint case since the variance of $\textbf{F}$ is greater than that of $\textbf{F}_{\widehat{\boldsymbol{\phi}}}$, which results in a greater allowance on the deviation. Comparing the results in Fig. \ref{fig:RequiredN} with that in Fig. \ref{fig:Asym_RSS} and Fig. \ref{fig:Asym_Joint}, we observe that the predicted $N$ for RSS-only case is accurate, in that at predicted $N$, the asymptotic RSS-only CRB in Fig. \ref{fig:Asym_RSS} well-approaches the exact CRB; on the other hand, the predicted $N$ for the joint case is in general more than enough. For example, when the $R_0/R = 0.3$, the predicted $N$ is 90; however, result in Fig. \ref{fig:Asym_Joint} indicates using 60 nodes already provides satisfying approximation accuracy.

\begin{figure}
\centering{
\includegraphics[width=0.6\columnwidth]{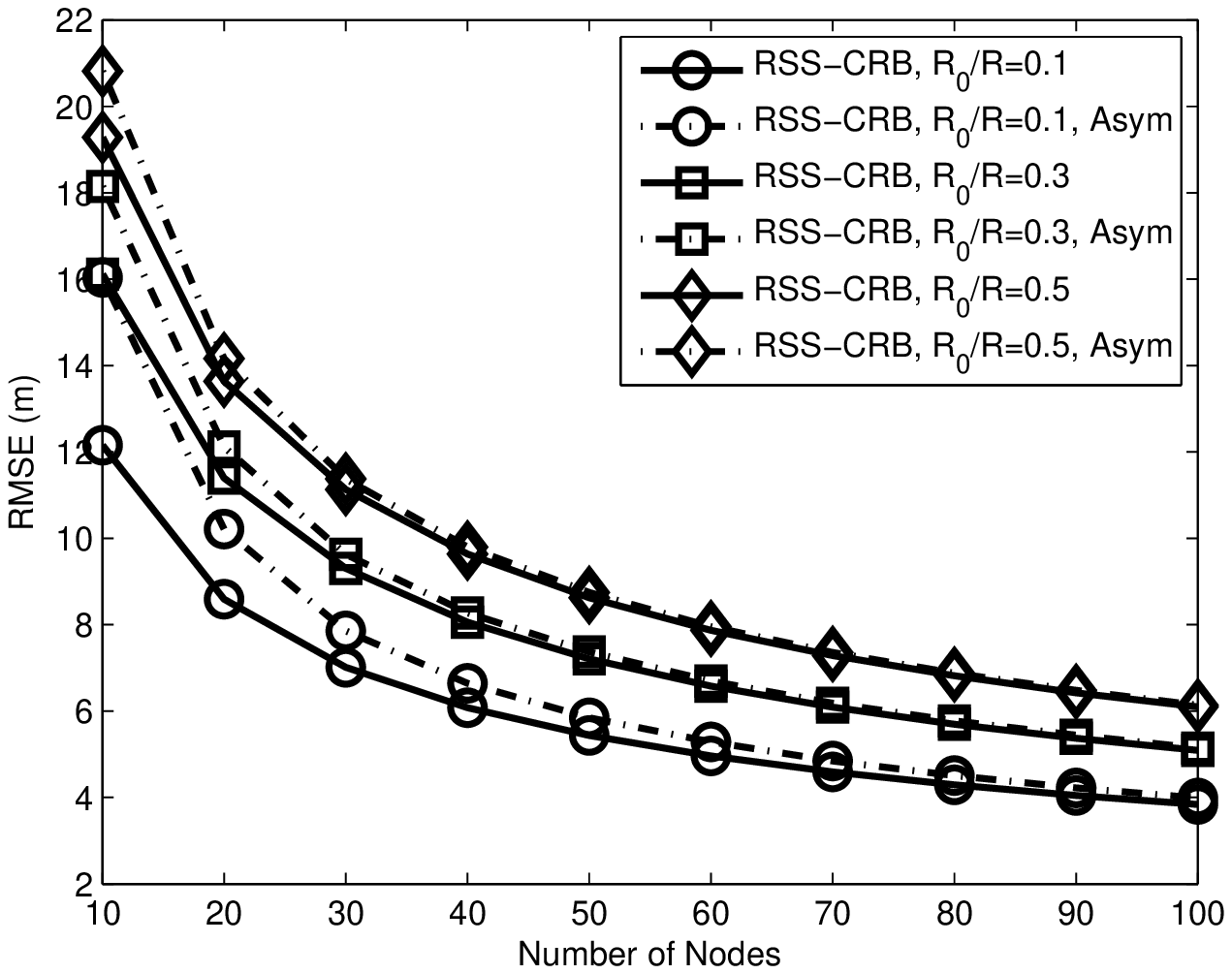}}
\caption{Results for exact and asymptotic RMSE of RSS-only CRB with varying number of CRs, for different $R_0$.}
\label{fig:Asym_RSS}
\end{figure}
\begin{figure}
\centering{
\includegraphics[width=0.6\columnwidth]{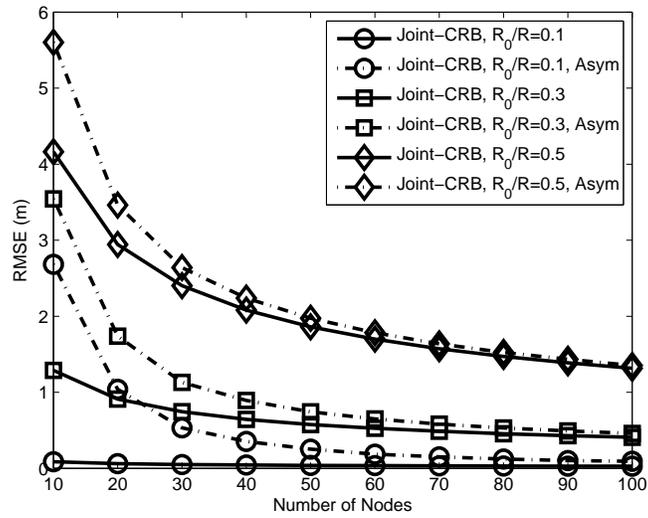}}
\caption{Results for exact and asymptotic RMSE of the joint CRB with varying number of CRs, for different $R_0$.}
\label{fig:Asym_Joint}
\end{figure}
\begin{figure}
\centering{
\includegraphics[width=0.6\columnwidth]{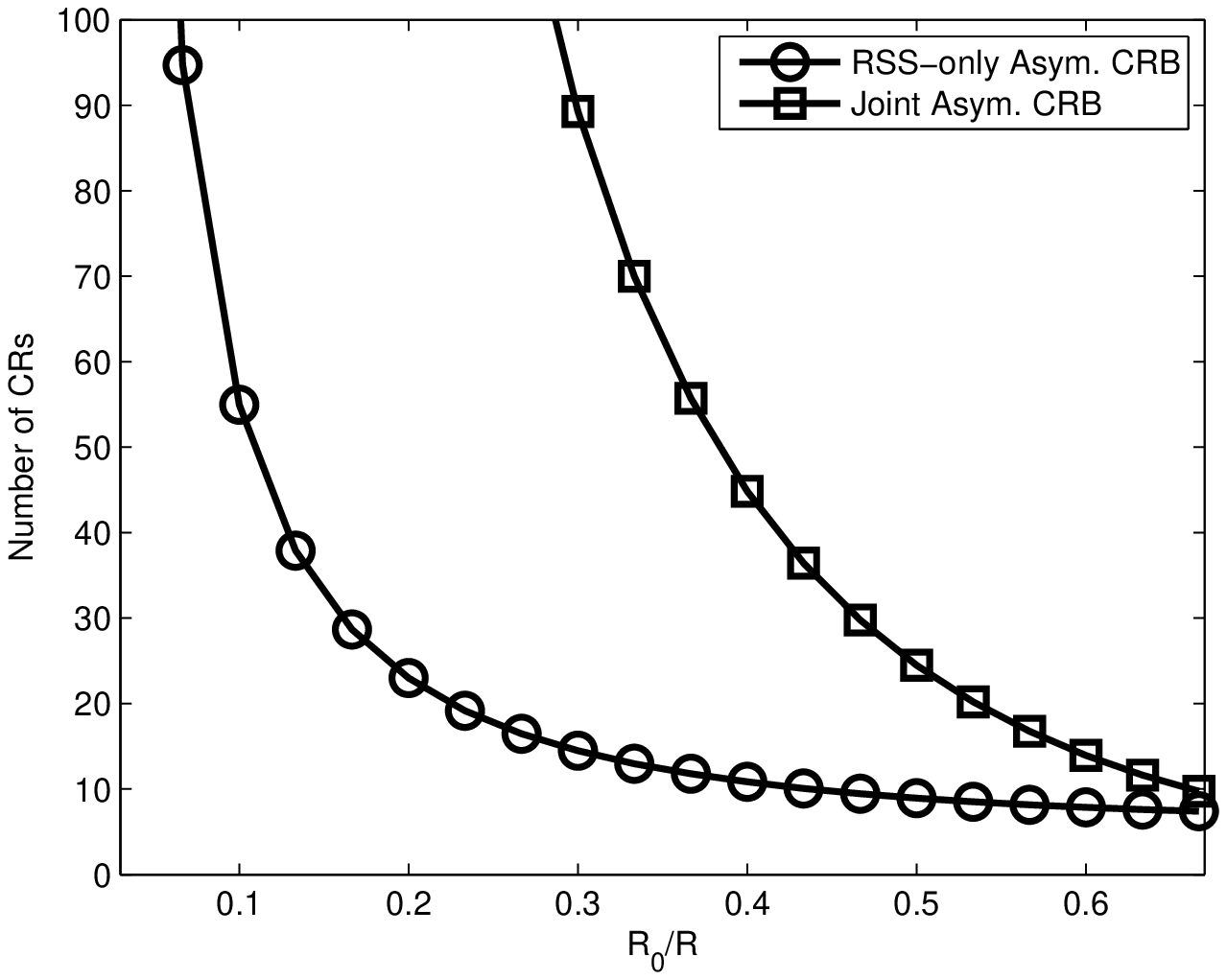}}
\caption{Required number of CRs predicted by the Chebyshev's inequality for the asymptotic CRBs to be accurate with varying $R_0$. For this scenario, $\delta_0 = \left\| \frac{1}{N} \E \left[ \textbf{F}_{\widehat{\boldsymbol{\phi}}} \right] \right\|_F$ for RSS-only case, $\delta_0 = 2 \left\| \frac{1}{N} \E \left[ \textbf{F} \right] \right\|_F$ for the joint case and $\eta = 0.15$.}
\label{fig:RequiredN}
\end{figure}
%

\section{Conclusion}
\label{sec:Conclusion}

In this paper we present a novel analytical framework to evaluate the achievable performance of joint RSS/DoA algorithms for localization of PU transmitters in CR networks, which explicitly considers the dependency of DoA measurement quality on instant RSS. We first derive the joint CRB for a fixed CR placement, for both optimal DoA measurement and the MUSIC algorithm, respectively. We then derive the asymptotic joint CRB for uniform CR placement and provide theoretical and numerical study of its approximation accuracy. Using the derived joint CRBs in conjunction with numerical simulations, the impact of various system, antenna array and channel variables on the joint CRB and several practical localization algorithms is investigated and quantified. Based on the above analysis, we provide guidelines in terms of number of nodes and antenna array settings required to achieve certain localization accuracy, which are useful for practical localization system design. It is observed that the achievable RMSE decreases for large shadowing standard derivation due to cooperative diversity, and is robust to correlated shadowing. With a small number of antennas, say 2 to 3, and sufficient number of samples, jointly using RSS and DoA measurements reduces the RMSE by $90 \%$ compared with using only RSS.

\appendices
\section{Derivation of $\textbf{F}_{\widehat{\boldsymbol{\phi}}}$ and $\textbf{F}_{\widehat{\boldsymbol{\theta}} | \widehat{\boldsymbol{\phi}}}$}
\label{Appendix1}

In this appendix we provide detailed derivations to obtain elements of $\textbf{F}_{\widehat{\boldsymbol{\phi}}}$ and $\textbf{F}_{\widehat{\boldsymbol{\theta}} | \widehat{\boldsymbol{\phi}}}$. For $\textbf{F}_{\widehat{\boldsymbol{\phi}}}$, starting from (\ref{RSS_FIM}), each element in $\textbf{F}_{\widehat{\boldsymbol{\phi}}}$ is given as
\begin{equation}\label{RSS_FIM_Expand1}
    \left\{ \textbf{F}_{\widehat{\boldsymbol{\phi}}} \right\}_{11} = \E_{\widehat{\boldsymbol{\phi}}} \left[ \frac{\partial (\widehat{\boldsymbol{\phi}} - \overline{\boldsymbol{\phi}})^{T}}{\partial x_P}  \boldsymbol{\Omega}_{\textbf{s}}^{-1} \frac{\partial (\widehat{\boldsymbol{\phi}} - \overline{\boldsymbol{\phi}})}{\partial x_P} + (\widehat{\boldsymbol{\phi}} - \overline{\boldsymbol{\phi}}) \boldsymbol{\Omega}_{\textbf{s}}^{-1} \frac{\partial^2 (\widehat{\boldsymbol{\phi}} - \overline{\boldsymbol{\phi}})}{\partial x_P^2}\right]
\end{equation}
\begin{equation}\label{RSS_FIM_Expand2}
\left\{ \textbf{F}_{\widehat{\boldsymbol{\phi}}} \right\}_{22} = \E_{\widehat{\boldsymbol{\phi}}} \left[ \frac{\partial (\widehat{\boldsymbol{\phi}} - \overline{\boldsymbol{\phi}})^{T}}{\partial y_P}  \boldsymbol{\Omega}_{\textbf{s}}^{-1} \frac{\partial (\widehat{\boldsymbol{\phi}} - \overline{\boldsymbol{\phi}})}{\partial y_P} + (\widehat{\boldsymbol{\phi}} - \overline{\boldsymbol{\phi}}) \boldsymbol{\Omega}_{\textbf{s}}^{-1} \frac{\partial^2 (\widehat{\boldsymbol{\phi}} - \overline{\boldsymbol{\phi}})}{\partial y_P^2}\right]
\end{equation}
\begin{equation}\label{RSS_FIM_Expand3}
\left\{ \textbf{F}_{\widehat{\boldsymbol{\phi}}} \right\}_{12} = \left\{ \textbf{F}_{\widehat{\boldsymbol{\phi}}} \right\}_{21} = \E_{\widehat{\boldsymbol{\phi}}} \left[ \frac{\partial (\widehat{\boldsymbol{\phi}} - \overline{\boldsymbol{\phi}})^{T}}{\partial x_P}  \boldsymbol{\Omega}_{\textbf{s}}^{-1} \frac{\partial (\widehat{\boldsymbol{\phi}} - \overline{\boldsymbol{\phi}})}{\partial y_P} + (\widehat{\boldsymbol{\phi}} - \overline{\boldsymbol{\phi}}) \boldsymbol{\Omega}_{\textbf{s}}^{-1} \frac{\partial^2 (\widehat{\boldsymbol{\phi}} - \overline{\boldsymbol{\phi}})}{\partial x_P \partial y_P}\right],
\end{equation}
where elements for all vectors of derivatives are given by $\frac{\partial (\widehat{\boldsymbol{\phi}}_n - \overline{\boldsymbol{\phi}}_n)}{\partial x_P} = \frac{10 \gamma}{\log 10} \frac{\Delta x_n}{d_n^2}$, $\frac{\partial (\widehat{\boldsymbol{\phi}}_n - \overline{\boldsymbol{\phi}}_n)}{\partial y_P} = \frac{10 \gamma}{\log 10} \frac{\Delta y_n}{d_n^2}$, $\frac{\partial^2 (\widehat{\boldsymbol{\phi}}_n - \overline{\boldsymbol{\phi}}_n)}{\partial x_P^2} = \frac{10 \gamma}{\log 10} \frac{\Delta y_n^2 - \Delta x_n^2}{d_n^4}$, $\frac{\partial^2 (\widehat{\boldsymbol{\phi}}_n - \overline{\boldsymbol{\phi}}_n)}{\partial y_P^2} = \frac{10 \gamma}{\log 10} \frac{\Delta x_n^2 - \Delta y_n^2}{d_n^4}$, $\frac{\partial^2 (\widehat{\boldsymbol{\phi}}_n - \overline{\boldsymbol{\phi}}_n)}{\partial x_P \partial x_P} = - \frac{20 \gamma}{\log 10} \frac{\Delta x_n \Delta y_n}{d_n^4}$. Applying the expectation in (\ref{RSS_FIM_Expand1}-\ref{RSS_FIM_Expand3}) with the derivatives results in (\ref{RSS_FIM_Elements}).

For $\textbf{F}_{\widehat{\boldsymbol{\theta}} | \widehat{\boldsymbol{\phi}}}$, starting from (\ref{DoA_FIM}), the second order derivatives of $g_{n}$ over $x_P$ is given by
\begin{eqnarray}\label{g1_2nd_Derivatives1}
  \frac{\partial^2 g_{n}}{\partial x_P^2} &=& - \frac{\Delta y_n}{\beta f(\widehat{\phi}_n) d_n^4} \left\{ \left[ \Delta x_n \sin (2 \widetilde{\theta}_n) + \Delta y_n \cos (2 \widetilde{\theta}_n) \right] (\widehat{\theta}_n - \theta_n)^2 \right. \nonumber \\
&& \left. + 2 \left[ \Delta x_n \cos^2 \widetilde{\theta}_n - \Delta y_n \sin (2 \widetilde{\theta}_n) \right] (\widehat{\theta}_n - \theta_n) - \Delta y_n \cos^2 \widetilde{\theta}_n \right\}
\end{eqnarray}
Note that here we use the general form $f(\widehat{\phi}_n)$ to represent the function of RSS in the DoA estimation error variance, instead of using specific functions such as $f_{CRB} (\widehat{\phi}_n)$ or $f_{MU} (\widehat{\phi}_n)$. Therefore the results in this appendix can be applied to derivation of $\textbf{F}_{\widehat{\boldsymbol{\theta}} | \widehat{\boldsymbol{\phi}}}$ of both assuming optimal DoA estimation or using MUSIC algorithm. When we take expectations of (\ref{g1_2nd_Derivatives1}), we apply the following property of taking expectation of a function $f(x,y)$ of two random variables $X$ and $Y$ with known joint PDF $p(x,y)$: $\E_{X,Y} \left[ f(x,y) \right] = \E_{Y} \left\{ \E_{X|Y} \left[ f(x,y) \right] \right\}$, and get
\begin{equation}\label{g1_Exp_Part1}
    \E_{\widehat{\boldsymbol{\theta}}, \widehat{\boldsymbol{\phi}}} \left[ \frac{\partial^2 g_{n}}{\partial x_P^2} \right] = \E_{\widehat{\boldsymbol{\phi}}} \left\{ \E_{\widehat{\boldsymbol{\theta}}|\widehat{\boldsymbol{\phi}}} \left[ \frac{\partial^2 g_{n}}{\partial x_P^2} \right] \right\} = \frac{\Delta y_n^2 \cos^2 \widetilde{\theta}_n}{\beta d_n^4} \E_{\widehat{\boldsymbol{\phi}}} \left[ \frac{1}{f(\widehat{\phi}_n)} \right] - \frac{\Delta y_n \left[ \Delta x_n \sin (2 \widetilde{\theta}_n) + \Delta y_n \cos (2 \widetilde{\theta}_n) \right]}{d_n^4 \cos^2 \widetilde{\theta}_n},
\end{equation}
Applying similar procedure to $\frac{\partial^2 g_{n}}{\partial x_P \partial y_P}$ and $\frac{\partial^2 g_{n}}{\partial y_P^2}$, we obtain
\begin{eqnarray}\label{g1_Exp_Part2}
\E_{\widehat{\boldsymbol{\theta}}, \widehat{\boldsymbol{\phi}}} \left[ \frac{\partial^2 g_{n}}{\partial x_P \partial y_P} \right] &=& - \frac{\Delta x_n \Delta y_n \cos^2 \widetilde{\theta}_n}{\beta d_n^4} \E_{\widehat{\boldsymbol{\phi}}} \left[ \frac{1}{f(\widehat{\phi}_n)} \right] + \frac{\left[ \Delta x_n \Delta y_n \cos (2 \widetilde{\theta}_n) + \frac{1}{2} (\Delta x_n^2 - \Delta y_n^2) \sin (2 \widetilde{\theta}_n) \right]}{d_n^4 \cos^2 \widetilde{\theta}_n} \nonumber \\
\E_{\widehat{\boldsymbol{\theta}}, \widehat{\boldsymbol{\phi}}} \left[ \frac{\partial^2 g_{n}}{\partial y_P^2} \right] &=& \frac{\Delta x_n^2 \cos^2 \widetilde{\theta}_n}{\beta d_n^4} \E_{\widehat{\boldsymbol{\phi}}} \left[ \frac{1}{f(\widehat{\phi}_n)} \right] - \frac{\Delta x_n \left[ \Delta x_n \cos (2 \widetilde{\theta}_n) - \Delta y_n \sin (2 \widetilde{\theta}_n) \right]}{d_n^4 \cos^2 \widetilde{\theta}_n}.
\end{eqnarray}
Now we consider different $f(\widehat{\phi}_n)$ given by CRB or MUSIC, which results in
\begin{subnumcases}{\E_{\widehat{\boldsymbol{\phi}}} \left[ \frac{1}{f(\widehat{\phi}_n)} \right]=}
    \frac{c_0 P_T e^{\sigma_s^2/(2 \epsilon)}}{d_n^{\gamma}}, & CRB, \label{Func_RSS_CRB} \\
    \frac{c_0 P_T e^{\sigma_s^2/(2 \epsilon)}}{d_n^{\gamma}} - \frac{P_M}{N_a} + \frac{P_M^2}{N_a^2} \E_{\widehat{\boldsymbol{\phi}}} \left[ \frac{1}{\widehat{\psi}_n + \frac{P_M}{N_a}} \right], & MUSIC,
\label{Func_RSS_MUSIC}
\end{subnumcases}
where the expectation in (\ref{Func_RSS_MUSIC}) is an expectation of inverse of shifted log-normal variable, which is not directly obtainable and need to be calculated through numerical methods.

The second order derivatives of $h_{n}$ over the PU coordinates are given by
\begin{eqnarray}\label{g2_2nd_Derivatives}
    \frac{\partial^2 h_{n}}{\partial x_P^2} &=& - \frac{\Delta y_n}{d_n^4 \cos^2 \widetilde{\theta}_n} \left[ \Delta y_n + \Delta x_n \sin (2 \widetilde{\theta}_n) \right] = \E_{\widehat{\boldsymbol{\theta}}, \widehat{\boldsymbol{\phi}}} \left[ \frac{\partial^2 h_{n}}{\partial x_P^2} \right] \nonumber \\
    \frac{\partial^2 h_{n}}{\partial y_P^2} &=& - \frac{\Delta x_n}{d_n^4 \cos^2 \widetilde{\theta}_n} \left[ \Delta x_n - \Delta y_n \sin (2 \widetilde{\theta}_n) \right] = \E_{\widehat{\boldsymbol{\theta}}, \widehat{\boldsymbol{\phi}}} \left[ \frac{\partial^2 h_{n}}{\partial y_P^2} \right] \nonumber \\
    \frac{\partial^2 h_{n}}{\partial x_P \partial y_P} &=& \frac{1}{d_n^4 \cos^2 \widetilde{\theta}_n} \left[ \Delta x_n \Delta y_n + \frac{1}{2} (\Delta x_n^2 - \Delta y_n^2) \sin (2 \widetilde{\theta}_n) \right] = \E_{\widehat{\boldsymbol{\theta}}, \widehat{\boldsymbol{\phi}}} \left[ \frac{\partial^2 h_{n}}{\partial x_P \partial y_P} \right],
\end{eqnarray}
where the second equality holds because the second order derivatives of $h_{n}$ are not related to random variables $\widehat{\boldsymbol{\theta}}$ or $\widehat{\boldsymbol{\phi}}$. Summation of corresponding elements in (\ref{g1_Exp_Part1}), (\ref{g1_Exp_Part2}) and (\ref{g2_2nd_Derivatives}) with some simplification will give us the elements of sub-FIM $\textbf{F}_{\widehat{\boldsymbol{\theta}} | \widehat{\boldsymbol{\phi}}}$.


\section{Proof of Theorem \ref{RSS_FIM_Theorem} and \ref{Joint_FIM_Theorem}}
\label{Appendix3}

We start by introducing the Chebyshev's inequality for random matrix, which will be the basic tool for the following proof.

\begin{theorem}
\emph{(Chebyshev's Inequality for Random Matrix)}
\label{Chebyshev}
Suppose $\textbf{X}$ is a random matrix with mean $\E \left[ \textbf{X} \right]$ and its second order moments exist. Then for any $\delta > 0$, we have
\begin{equation}\label{Chebyshev_Ineq}
    \text{Pr} \left\{ \left\| \textbf{X} - \E \left[ \textbf{X} \right] \right\| > \delta \right\} < \frac{\E \left[ \left\| \textbf{X} - \E \left[ \textbf{X} \right] \right\|^2 \right]}{\delta^2},
\end{equation}
where $\left\| \cdot \right\|$ denotes any matrix norm.
\end{theorem}

\begin{proof}
From the Markov's inequality which states that if $X$ is a nonnegative random variable, and $\E \left[ X \right]$ is its mean, then for any $\delta > 0$, we have $\text{Pr} \{ X > \delta \} < \frac{\E \left[ X \right]}{\delta}$. Therefore, for scalar random variable $\left\| \textbf{X} - \E \left[ \textbf{X} \right] \right\|^2$, we have $\text{Pr} \left\{ \left\| \textbf{X} - \E \left[ \textbf{X} \right] \right\|^2 > \delta^2 \right\} < \frac{\E \left[ \left\| \textbf{X} - \E \left[ \textbf{X} \right] \right\|^2 \right]}{\delta^2}$. Taking square root of the expression on the left-hand-side finishes the proof.
\end{proof}

To prove Theorem \ref{RSS_FIM_Theorem}, we observe from the Chebyshev's inequality that we need to compute $\E \left[ \left\| \textbf{X} - \E \left[ \textbf{X} \right] \right\|^2_F \right]$ for $\textbf{X} = \frac{1}{N} \textbf{F}_{\widehat{\boldsymbol{\phi}}}$, where the Frobenius norm is given by $\parallel \textbf{X} \parallel_{F} = \sqrt {\text{Tr} (\textbf{X}^T \textbf{X})}$. As a result, we need to evaluate the following expression
\begin{equation}\label{RSS_FIM_Evaluate}
  \E \left\{ \text{Tr} \left[ \left( \frac{1}{N} \textbf{F}_{\widehat{\boldsymbol{\phi}}} - \frac{1}{N} \E \left[ \textbf{F}_{\widehat{\boldsymbol{\phi}}} \right] \right)^T \left( \frac{1}{N} \textbf{F}_{\widehat{\boldsymbol{\phi}}} - \frac{1}{N} \E \left[ \textbf{F}_{\widehat{\boldsymbol{\phi}}} \right] \right) \right] \right\} = \frac{1}{N^2} \text{Tr} \left\{ \E \left[ \left( \textbf{F}_{\widehat{\boldsymbol{\phi}}} - \E \textbf{F}_{\widehat{\boldsymbol{\phi}}} \right)^T \left( \textbf{F}_{\widehat{\boldsymbol{\phi}}} - \E \textbf{F}_{\widehat{\boldsymbol{\phi}}} \right) \right] \right\}.
\end{equation}
The covariance of $\textbf{F}_{\widehat{\boldsymbol{\phi}}}$ appeared in the right-hand-side of (\ref{RSS_FIM_Evaluate}) can be derived from (\ref{RSS_FIM_Rewrite}) as
\begin{equation}\label{RSS_FIM_Covariance}
  \E \left[ \left( \textbf{F}_{\widehat{\boldsymbol{\phi}}} - \E \left[ \textbf{F}_{\widehat{\boldsymbol{\phi}}} \right] \right)^T \left( \textbf{F}_{\widehat{\boldsymbol{\phi}}} - \E \left[ \textbf{F}_{\widehat{\boldsymbol{\phi}}} \right] \right) \right] = \frac{\epsilon^2 \gamma^4 N}{\sigma_s^4} \left[ \frac{1}{2 R^2 R_0^2} - \frac{\log^2(R/R_0)}{(R^2-R_0^2)^2} \right] \textbf{I}_2.
\end{equation}
From (\ref{RSS_FIM_Evaluate}) and (\ref{RSS_FIM_Covariance}) we obtain that
\begin{equation}\label{RSS_FIM_Moment}
  \E \left[ \left\| \frac{1}{N} \textbf{F}_{\widehat{\boldsymbol{\phi}}} - \frac{1}{N} \E \left[  \textbf{F}_{\widehat{\boldsymbol{\phi}}} \right] \right\|^2_F \right] = \frac{2 \epsilon^2 \gamma^4 }{\sigma_s^4 N} \left[ \frac{1}{2 R^2 R_0^2} - \frac{\log^2(R/R_0)}{(R^2-R_0^2)^2} \right].
\end{equation}
The application of Theorem \ref{Chebyshev} on $\left\| \frac{1}{N} \textbf{F}_{\widehat{\boldsymbol{\phi}}} - \frac{1}{N} \E \left[  \textbf{F}_{\widehat{\boldsymbol{\phi}}} \right] \right\|_F$ with result of (\ref{RSS_FIM_Moment}) proves Theorem \ref{RSS_FIM_Theorem}.

To prove Theorem \ref{Joint_FIM_Theorem}, we observe from the Chebyshev's inequality that we need to compute $\E \left[ \left\| \textbf{X} - \E \left[ \textbf{X} \right] \right\|^2_F \right]$ for $\textbf{X} = \frac{1}{N} \textbf{F}$, which requires evaluation of the following expression
\begin{equation}\label{Joint_FIM_Evaluate}
    \E \left\{ \text{Tr} \left[ \left( \frac{1}{N} \textbf{F} - \frac{1}{N} \E \left[ \textbf{F} \right] \right)^T \left( \frac{1}{N} \textbf{F} - \frac{1}{N} \E \left[ \textbf{F} \right] \right) \right] \right\} = \frac{1}{N^2} \text{Tr} \left\{ \E \left[ \left( \textbf{F} - \E \left[ \textbf{F} \right] \right)^T \left( \textbf{F} - \E \left[ \textbf{F} \right] \right) \right] \right\}.
\end{equation}
The covariance of $\textbf{F} = \textbf{F}_{\widehat{\boldsymbol{\phi}}} + \textbf{F}_{\widehat{\boldsymbol{\theta}} | \widehat{\boldsymbol{\phi}}}$ appeared on the right-hand-side of (\ref{Joint_FIM_Evaluate}) can be derived as
\begin{equation}\label{Joint_FIM_Covariance}
    \E \left[ \left( \textbf{F} - \E \left[ \textbf{F} \right] \right)^T \left( \textbf{F} - \E \left[ \textbf{F} \right] \right) \right]
= N \left\{ \frac{\epsilon^2 \gamma^4}{2 \sigma_s^4 R^2 R_0^2}  + \frac{1}{2} \E \left[ f_n ^2\right] - \frac{1}{4} \left[ \E \left[ f_n \right] + \frac{2 \epsilon \gamma^2 \log(R/R_0)}{\sigma_s^2 (R^2 - R_0^2)}\right]^2  \right\} \textbf{I}_2,
\end{equation}
where $f_n$, $\E \left[ f_n \right]$ and $\E \left[ f_n^2 \right]$ are given by (\ref{Theorem2_Eqn2}). We skip the derivations to obtain (\ref{Joint_FIM_Covariance}) which include tedious manipulations of (\ref{RSS_FIM_Rewrite}) and (\ref{DoA_FIM_Rewrite}). From (\ref{Joint_FIM_Evaluate}) and (\ref{Joint_FIM_Covariance}) we obtain that
\begin{equation}\label{Joint_FIM_Moment}
  \E \left[ \left\| \frac{1}{N} \textbf{F} - \frac{1}{N} \E \left[  \textbf{F} \right] \right\|^2_F \right] = \frac{1}{N} \left\{ \frac{\epsilon^2 \gamma^4}{\sigma_s^4 R^2 R_0^2}  + \E \left[ f_n ^2\right] - \frac{1}{2} \left[ \E \left[ f_n \right] + \frac{2 \epsilon \gamma^2 \log(R/R_0)}{\sigma_s^2 (R^2 - R_0^2)}\right]^2  \right\} .
\end{equation}
The application of Theorem \ref{Chebyshev} on $\left\| \frac{1}{N} \textbf{F} - \frac{1}{N} \E \left[  \textbf{F} \right] \right\|_F$ with result of (\ref{Joint_FIM_Moment}) proves Theorem \ref{Joint_FIM_Theorem}.


\bibliographystyle{IEEEtran}
\bibliography{IEEEabrv,reference}

\begin{thebibliography}{10}
\providecommand{\url}[1]{#1}
\csname url@samestyle\endcsname
\providecommand{\newblock}{\relax}
\providecommand{\bibinfo}[2]{#2}
\providecommand{\BIBentrySTDinterwordspacing}{\spaceskip=0pt\relax}
\providecommand{\BIBentryALTinterwordstretchfactor}{4}
\providecommand{\BIBentryALTinterwordspacing}{\spaceskip=\fontdimen2\font plus
\BIBentryALTinterwordstretchfactor\fontdimen3\font minus
  \fontdimen4\font\relax}
\providecommand{\BIBforeignlanguage}[2]{{%
\expandafter\ifx\csname l@#1\endcsname\relax
\typeout{** WARNING: IEEEtran.bst: No hyphenation pattern has been}%
\typeout{** loaded for the language `#1'. Using the pattern for}%
\typeout{** the default language instead.}%
\else
\language=\csname l@#1\endcsname
\fi
#2}}
\providecommand{\BIBdecl}{\relax}
\BIBdecl

\bibitem{Mitola1999}
J.~{Mitola} and G.~{Maguire}, ``Cognitive radio: making software radios more
  personal,'' \emph{IEEE Personal Commun.}, vol.~6, no.~4, pp. 13--18, Aug
  1999.

\bibitem{Haykin2005}
S.~{Haykin}, ``Cognitive radio: brain-empowered wireless communications,''
  \emph{{IEEE} J. Sel. Areas Commun.}, vol.~23, no.~2, pp. 201--220, Feb. 2005.

\bibitem{Celebi2007a}
H.~Celebi and H.~Arslan, ``Utilization of location information in cognitive
  wireless networks,'' \emph{{IEEE} Wireless Commun. Mag.}, vol.~14, no.~4, pp.
  6 --13, Aug. 2007.

\bibitem{Patwari2005}
N.~{Patwari}, J.~{Ash}, S.~{Kyperountas}, A.~{Hero}, R.~{Moses}, and
  N.~{Correal}, ``Locating the nodes: cooperative localization in wireless
  sensor networks,'' \emph{{IEEE} Signal Process. Mag.}, vol.~22, no.~4, pp.
  54--69, July 2005.

\bibitem{Kaplan2005}
E.~{Kaplan} and C.~{Hegarty}, ``Understanding gps: Principles and
  applications,'' \emph{MA: Artech House}, 2005.

\bibitem{Wang2011}
J.~{Wang}, P.~{Urriza}, Y.~{Han}, and D.~{Cabric}, ``Weighted centroid
  localization algorithm: theoretical analysis and distributed
  implementation,'' \emph{{IEEE} Trans. Wireless Commun.}, vol.~10, no.~10, pp.
  3403--3413, Oct. 2011.

\bibitem{Gavish1992}
M.~{Gavish} and A.~{Weiss}, ``Performance analysis of bearing-only target
  location algorithms,'' \emph{{IEEE} Trans. Aerosp. Electron. Syst.}, vol.~28,
  no.~3, pp. 817--828, Jul. 1992.

\bibitem{Stansfield1947}
R.~{Stansfield}, ``Statistical theory of d.f. fixing,'' \emph{J. of the Inst.
  of Electr. Eng.}, vol.~94, no.~15, pp. 762 --770, Mar. 1947.

\bibitem{Stoica1989}
P.~{Stoica} and N.~{Arye}, ``Music, maximum likelihood, and cramer-rao bound,''
  \emph{{IEEE} Trans. Acoust., Speech, Signal Process.}, vol.~37, no.~5, pp.
  720 --741, May 1989.

\bibitem{Patwari2008}
N.~Patwari and P.~Agrawal, ``Effects of correlated shadowing: connectivity,
  localization, and rf tomography,'' in \emph{Proc. ACM/IEEE IPSN}, Apr. 2008,
  pp. 82 --93.

\bibitem{Gustafsson2005}
F.~Gustafsson and F.~Gunnarsson, ``Mobile positioning using wireless networks:
  possibilities and fundamental limitations based on available wireless network
  measurements,'' \emph{{IEEE} Signal Process. Mag.}, vol.~22, no.~4, pp.
  41--53, July 2005.

\bibitem{Ash2008}
J.~Ash and R.~Moses, ``On the relative and absolute positioning errors in
  self-localization systems,'' \emph{{IEEE} Trans. Signal Process.}, vol.~56,
  no.~11, pp. 5668--5679, Nov. 2008.

\bibitem{Seow2008}
C.~{Seow} and S.~{Tan}, ``Non-line-of-sight localization in multipath
  environments,'' \emph{{IEEE} Trans. Mobile Comput.}, vol.~7, no.~5, pp.
  647--660, May 2008.

\bibitem{Vaghefi2010}
R.~Vaghefi, M.~Gholami, and E.~Strom, ``Bearing-only target localization with
  uncertainties in observer position,'' in \emph{Proc. IEEE PIMRC}, Sept. 2010,
  pp. 238 --242.

\bibitem{Fu2009}
Y.~{Fu} and Z.~{Tian}, ``Cramer-rao bounds for hybrid toa/doa-based location
  estimation in sensor networks,'' \emph{{IEEE} Signal Process. Lett.},
  vol.~16, no.~8, pp. 655--658, Aug. 2009.

\bibitem{Penna2011}
F.~Penna and D.~Cabric, ``Bounds and tradeoffs for cooperative doa-only
  localization of primary users,'' in \emph{Proc. IEEE GLOBECOM}, Dec. 2011,
  pp. 1--5.

\bibitem{Qi2006}
Y.~Qi, H.~Kobayashi, and H.~Suda, ``Analysis of wireless geolocation in a
  non-line-of-sight environment,'' \emph{{IEEE} Trans. Wireless Commun.},
  vol.~5, no.~3, pp. 672--681, Mar. 2006.

\bibitem{Schmidt1986}
R.~Schmidt, ``Multiple emitter location and signal parameter estimation,''
  \emph{{IEEE} Trans. Antennas Propag.}, vol.~34, no.~3, pp. 276--280, Mar.
  1986.

\bibitem{Roy1989}
R.~Roy and T.~Kailath, ``Esprit-estimation of signal parameters via rotational
  invariance techniques,'' \emph{{IEEE} Trans. Acoust., Speech, Signal
  Process.}, vol.~37, no.~7, pp. 984--995, Jul. 1989.

\bibitem{Proakis2001}
J.~G. {Proakis}, \emph{Digital Communications}.\hskip 1em plus 0.5em minus
  0.4em\relax New York, NY, USA: McGraw-Hill, 2001.

\bibitem{Wang2012}
J.~Wang and D.~Cabric, ``A cooperative doa-based algorithm for localization of
  multiple primary-users in cognitive radio networks,'' in \emph{accepted for
  publication in Proc. IEEE GLOBECOM}, 2012.

\bibitem{Taaghol1997}
P.~Taaghol and R.~Tafazolli, ``Correlation model for shadow fading in
  land-mobile satellite systems,'' \emph{Electron. Lett.}, vol.~33, no.~15, pp.
  1287 --1289, Jul 1997.

\end{thebibliography}

\end{document}